\documentclass{article}
\usepackage[utf8]{inputenc}

\newcommand{\knote}[1]{{\color{red}[{\tiny Karthik: \bf #1}]\marginpar{\color{red}*}}}

\usepackage[T1]{fontenc}
\usepackage{soul} 
\usepackage{phfqit}
\usepackage{proba} 
\usepackage{amssymb}
\usepackage{amsmath}
\usepackage{amsthm}
\usepackage{amsfonts}
\usepackage{mathtools}
\usepackage{xspace}
\usepackage{bm}
\usepackage{nicefrac}
\usepackage{commath}
\usepackage{bbm}
\usepackage{boxedminipage}
\usepackage{xparse}
\usepackage{xcolor}
\usepackage{float}
\usepackage{multirow}
\usepackage{graphicx}
\usepackage{caption}
\usepackage{subcaption}
\usepackage{ifthen}
\usepackage{algorithmic}
\usepackage{algorithm}
\usepackage{colortbl} 
\usepackage{fancyhdr}   
\usepackage{hyperref}
    \hypersetup{ colorlinks=true, linkcolor=blue, filecolor=magenta, urlcolor=blue,
    citecolor=magenta,}
\usepackage{fullpage}
\usepackage[utf8]{inputenc}
\usepackage{environ}
\usepackage{mdframed}
\usepackage{cleveref}
\usepackage[short]{optidef} 
\usepackage{enumitem}
\usepackage{thm-restate}
\usepackage{tikz}
\usetikzlibrary{shapes}
\usetikzlibrary{positioning}
\usetikzlibrary{fit}
\usetikzlibrary{graphs,graphs.standard}

\newtheorem{proposition}{Proposition}
\newtheorem{lemma}{Lemma}
\newtheorem{remark}{Remark}
\newtheorem{theorem}{Theorem}
\newtheorem{corollary}{Corollary}[lemma]
\newtheorem{thmcorollary}{Corollary}[theorem]

\newtheorem{conjecture}{Conjecture}

\newtheorem{claim}{Claim}

\newtheorem{question}{Question}

\NewEnviron{problem}[1]{%
	\begin{center}\fbox{\parbox{6in}{%
				{\centering\scshape #1\par}%
				\parskip=1ex
				\everypar{\hangindent=1em}%
				\BODY
}}\end{center}}

\newcommand{\nnreal}{\ensuremath{\mathbb{R}_{\geq 0}}\xspace}

\newcommand{\spanfunc}{\ensuremath{\mathsf{span}}\xspace}
\newcommand{\row}{\ensuremath{\mathsf{row}}\xspace}
\newcommand{\Rows}{\ensuremath{\mathsf{Rows}}\xspace}
\title{Polyhedral Aspects of Feedback Vertex Set \\and Pseudoforest Deletion Set\thanks{Corresponding Author: Karthekeyan Chandrasekaran, Affiliation: University of Illinois, Urbana-Champaign, E-mail Address: karthe@illinois.edu.}}
\author{Karthekeyan Chandrasekaran\thanks{University of Illinois, Urbana-Champaign, Email: \{karthe, smkulka2\}@illinois.edu. Supported in part by NSF grants CCF-1814613 and CCF-1907937.}
\and Chandra Chekuri\thanks{University of Illinois, Urbana-Champaign, Email: chekuri@illinois.edu. Supported in part by NSF grants CCF-1910149 and CCF-1907937.}
\and 
Samuel Fiorini\thanks{Universit\'{e} libre de Bruxelles, Email: sfiorini@ulb.ac.de}
\and Shubhang Kulkarni\footnotemark[2]
\and Stefan Weltge\thanks{Technische Universit\"at M\"unchen, Email: weltge@tum.de. Supported in part by  Deutsche Forschungsgemeinschaft (DFG, German Research Foundation), project number 451026932.}
}
\date{}

\newcommand{\strongdensitypolyhedron}{P_{\text{SD}}}
\newcommand{\weakdensitypolyhedron}{P_{\text{WD}}}
\newcommand{\cyclecoverpolyhedron}{P_{\text{cycle-cover}}}
\newcommand{\twopseudotreecoverpolyhedron}{P_{\text{2-PT-cover}}}
\newcommand{\orientationpolyhedron}{P_{\text{orient}}}
\newcommand{\projectedorientationpolyhedron}{Q_{\text{orient}}}
\newcommand{\distancebasedcyclecoverpolyhedron}{P_{\text{dist-cycle-cover}}}
\newcommand{\projecteddistancebasedcyclecoverpolyhedron}{Q_{\text{dist-cycle-cover}}}
\newcommand{\weakdensitypolyhedronforsubgraphs}{\ensuremath{{P}_{\text{WD-Subgraphs}}}\xspace}
\begin{document}
\maketitle
\begin{abstract}
  We consider the feedback vertex set problem in undirected graphs (FVS). The input to FVS is an undirected graph $G=(V,E)$ with non-negative vertex costs. The goal is to find a  minimum cost subset of vertices $S \subseteq V$ such that $G-S$ is acyclic. FVS is a well-known NP-hard problem and does not admit a $(2-\epsilon)$-approximation for any fixed $\epsilon > 0$ assuming the Unique Games Conjecture.   There are combinatorial $2$-approximation algorithms \cite{Bafna-Berman-Fujito95,BG96} and also primal-dual based $2$-approximations \cite{CHUDAK1998111,Fuj-matroid-FVS}. Despite the existence of these algorithms for several decades, there is no known polynomial-time solvable LP relaxation for FVS with a provable integrality gap of at most $2$. More recent work \cite{chekuri-madan16} developed a polynomial-sized LP relaxation for a more general problem, namely Subset FVS, and showed that its integrality gap is at most $13$ for Subset FVS, and hence also for FVS.

  Motivated by this gap in our knowledge, we undertake a polyhedral study of FVS and related problems. In this work, we formulate new integer linear programs (ILPs) for FVS whose LP-relaxation can be solved in polynomial time, and whose integrality gap is at most $2$. The new insights in this process also enable us to prove that the formulation in \cite{chekuri-madan16} has an integrality gap of at most $2$ for FVS. Our results for FVS are inspired by new formulations and polyhedral results for the closely-related pseudoforest deletion set problem (PFDS). Our formulations for PFDS are in turn inspired by a connection to the densest subgraph problem. We also conjecture an extreme point property for a LP-relaxation for FVS, and give evidence for the conjecture via a corresponding result for PFDS.
\end{abstract}

\newpage
\tableofcontents
\newpage

\section{Introduction}
\label{sec:intro}
We consider the feedback vertex set problem in undirected graphs. For
a graph $G=(V, E)$, a subset $U\subseteq V$ is a \emph{feedback vertex
  set} if $G-U$ is acyclic; in other words $U$ is a hitting set for
the cycles of $G$. In the feedback vertex set problem (FVS), the input
is an undirected graph $G=(V, E)$ with non-negative vertex costs
$c: V\rightarrow \R_{\ge 0}$ and the goal is to find a feedback vertex
set of minimum cost, i.e.,
$\min\left\{\sum_{u\in U} c(u): U\subseteq V\ \text{and $U$ is a feedback
  vertex set for $G$}\right\}$.  FVS is a fundamental vertex deletion
problem in the field of combinatorial optimization and appears in
Karp's list of 21 NP-Complete problems (as a decision version).  We
consider the approximability of FVS, in particular via polynomial-time
solvable polyhedral formulations. First, we mention some known lower
bounds. Via a simple reduction, the well-known weighted vertex cover
problem can be reduced to FVS in an approximation preserving
fashion. 
Hence, known results for the  vertex cover problem imply that there is no
polynomial-time $(2-\epsilon)$-approximation for every fixed constant
$\epsilon>0$ under the Unique Games Conjecture (UGC) \cite{KhotR08} and that there is no polynomial-time $(1.36-\epsilon)$-approximation for every fixed constant $\epsilon>0$ under the $P \neq NP$ assumption \cite{DinurS05}.

Approximation algorithms for the unweighted version of FVS (UFVS), i.e., when all
vertex costs are one, have been studied since 1980s.  An
$O(\log n)$-approximation for UFVS is implicit in the work of Erd\"os
and P\'osa \cite{ErdosP62}; Monien and Schulz \cite{MonienS81} seem to be the first ones
to explicitly study the problem, and obtained an
$O(\sqrt{\log n})$-approximation for UFVS. Bar-Yehuda, Geiger, Naor, and Roth \cite{BarYehudaGNR98} improved
the ratio for UFVS to $4$, and also noted that an $O(\log
n)$-approximation holds for FVS. Soon after that,
Bafna, Berman, and Fujito \cite{Bafna-Berman-Fujito95}, and independently Becker and Geiger \cite{BG96},
obtained $2$-approximation algorithms for FVS. While \cite{Bafna-Berman-Fujito95}
explicitly uses the local-ratio terminology, \cite{BG96} describes the algorithm
in a purely combinatorial fashion.
Although the algorithm in \cite{Bafna-Berman-Fujito95} is described via the
local-ratio method, the underlying LP relaxation is not obvious.
As observed in \cite{BarYehudaGNR98}, the natural hitting set LP relaxation for FVS
has an integrality gap of $\Theta(\log n)$.
Chudak, Goemans, Hochbaum, and Williamson \cite{CHUDAK1998111} described
exponential sized integer linear programming (ILP) formulations for FVS, and showed that the algorithms in \cite{Bafna-Berman-Fujito95,BG96} can be viewed as primal-dual algorithms
with respect to the LP relaxations of these formulations. This also established that
these LP relaxations have an integrality gap of at most $2$. 
Fujito \cite{Fuj-matroid-FVS} considered a unified generalization of FVS and vertex cover, namely matroidal FVS, and gave a $O(\log n)$-approximation for matroidal FVS via connections to submodular set cover. In the same work, Fujito formulated an exponential 
sized ILP for matroidal FVS, and designed a primal-dual algorithm with respect to its LP relaxation. He proved that the algorithm has an approximation
ratio of $2$ for a certain family of matroids. When specialized to FVS, we note that the ILP and the resulting primal-dual algorithm
in \cite{Fuj-matroid-FVS} is slightly different 
from the one in \cite{CHUDAK1998111} although they both yield $2$-approximations. 

We recall the integer linear program (ILP) formulation and the
associated LP-relaxation from \cite{CHUDAK1998111}. For a graph $G=(V,
E)$ and a subset $S\subseteq V$, let $E[S]:=\left\{\left\{u,
v\right\}\in E: u, v\in S\right\}$ and $G[S]:=(S, E[S])$, and for a
vertex $v\in V$, let $d_S(v)$ denote the degree of $v$ in $G[S]$. We
will denote the following polyhedron as the strong density polyhedron
and the constraints describing the polyhedron as strong density
constraints\footnote{The use of the ``density'' terminology will be clear later when we discuss
the connection to the densest subgraph problem and its LP relaxation.}:
\begin{align}
    \strongdensitypolyhedron(G)
    &:=
    \left\{ x\in \R^{V}_{\ge 0}:\ \sum \nolimits_{u\in S}(d_S(u)-1)x_u\geq |E[S]| - |S| + 1 \ \forall S\text{ such that } E[S]\not = \emptyset
\right\}.
\end{align}
Chudak, Goemans, Hochbaum, and Williamson proved that strong density constraints are valid for FVS and that
the following is an integer linear programming formulation for FVS:
\begin{align}
\min\left\{\sum \nolimits_{u\in V}c_u x_u:\ x\in \strongdensitypolyhedron(G)\cap \Z^V\right\}. \tag{FVS-IP: SD} \label{FVS-IP: SD}
\end{align}
They interpreted the local-ratio algorithm from \cite{Bafna-Berman-Fujito95} as a
primal-dual algorithm with respect to\footnote{An $\alpha$-approximation with respect to an LP $\min\{c^Tx: Ax\le b\}$ is an algorithm that returns an integral solution $z^*$ satisfying $Ax\le b$ such that $c^Tz^*\le \alpha \text{OPT}_{\text{LP}}$, where $\text{OPT}_{\text{LP}}$ is the optimum objective value of the LP.} the LP-relaxation of (\ref{FVS-IP: SD}) that we explicitly write below:
\begin{align}
    \min\left\{\sum \nolimits_{u\in V}c_u x_u:\ x\in \strongdensitypolyhedron(G)\right\}. \tag{FVS-LP: SD} \label{FVS-LP: SD}
\end{align}

It is, however, not known whether (\ref{FVS-LP: SD}) can be solved in
polynomial time. Equivalently, it is open to design a polynomial-time
separation oracle for the family of strong density constraints.
Moreover, there has been no other polynomial-time solvable LP
relaxation through which one could obtain a $2$-approximation for FVS.
This status leads to the following natural question:

\begin{question}
  \label{q:intro}
  Does there exist an ILP formulation for FVS whose LP-relaxation can be solved in polynomial time and has integrality gap at most $2$?
\end{question}

The lack of solvable LP relaxations for FVS with small integrality gap has also been a stumbling
block for the design of approximation algorithms for a generalization
of FVS called the subset feedback vertex set problem (Subset-FVS) \cite{EvenNSZ00}: the input is a
vertex-weighted graph $G$ and a terminal set $T \subseteq V$, and the goal is
to remove a minimum cost subset of vertices $S$ to ensure that $G-S$ has
no cycle containing a terminal $t \in T$. There is an $8$-approximation for
Subset-FVS \cite{EvenNZ00} and this is based on a complex algorithm that combines 
combinatorial and LP-based techniques. Chekuri and Madan \cite{chekuri-madan16} formulated a polynomial-sized integer linear
program for Subset-FVS and showed that the integrality gap of its
LP-relaxation is at most $13$. They explicitly raised the question of
whether the integrality gap of their formulation is better for FVS; in fact
it is open whether their formulation's integrality gap is at most $2$ for Subset-FVS.

In this work, we provide an affirmative answer to
Question~\ref{q:intro} by undertaking a polyhedral study of FVS and a
closely related problem, namely the pseudoforest deletion set problem
that will be described shortly. In addition to formulating
polynomial-time solvable LP relaxations with small integrality gap, it
is also of interest to find algorithms that can round fractional
solutions to the LP relaxations, and to understand properties of extreme points of
the corresponding polyhedra. Based on several interrelated technical
considerations, we conjecture the following property for
the strong density polyhedron:

\begin{conjecture}\label{conj:strong-density-extreme-point}
Let $G=(V, E)$ be a graph that contains a cycle. For every extreme
point $x$ of the polyhedron $\strongdensitypolyhedron(G)$, there
exists a vertex $u\in V$ such that $x_u\ge 1/2$.
\end{conjecture}
This conjecture would lead to an alternative proof that the
integrality gap of the LP-relaxation (\ref{FVS-LP: SD}) is at most $2$
via iterative rounding (instead of the primal-dual technique).
Although we were unable to resolve this conjecture for the strong
density polyhedron, we were able to show that a variant of the
conjecture holds for a \emph{weak density polyhedron} (to be defined
later). The weak density polyhedron closely resembles the strong
density polyhedron and is associated with an ILP formulation of the
pseudoforest deletion set problem.
We also discuss extreme point properties for other formulations later in the paper.

\paragraph{Pseudoforest Deletion Set Problem (PFDS).} A connected graph is a \emph{pseudotree} if
it has exactly one cycle; in other words there is an edge whose removal
results in a spanning tree. A graph is a \emph{pseudoforest} if every connected component is either acyclic or a pseudotree.  For a graph $G=(V, E)$, a
subset $U\subseteq V$ is a \emph{pseudoforest deletion set} if $G-U$
is a pseudoforest. In the pseudoforest deletion set problem (PFDS),
the input is an undirected graph $G=(V, E)$ with non-negative vertex
costs $c: V\rightarrow \R_{\ge 0}$ and the goal is to find a
pseudoforest deletion set of minimum cost, i.e., $\min \left\{\sum_{u\in U}
c(u): U\subseteq V\ \text{and $U$ is a pseudoforest deletion set for
  $G$}\right\}$. Intuitively, one can see that PFDS is closely related to FVS: 
we note that a feasible solution for FVS in a given graph $G$ is a feasible
solution to the PFDS instance on $G$. Also, finding an FVS in a graph that is
a pseudoforest is easy: for each connected component that is a pseudotree, we remove
the cheapest vertex in its unique cycle. 
PFDS and FVS are special cases of the more general $\ell$-pseudoforest deletion problem that was introduced in \cite{PhilipRS18} from the perspective of parameterized algorithms (FVS corresponds to $\ell=0$ and PFDS to $\ell=1$). Lin, Feng, Fu, and Wang \cite{LFFW19} studied approximation algorithms for $\ell$-pseudoforest deletion problem. In this paper, we restrict attention to PFDS and FVS and do not discuss the more general $\ell$-pseudoforest deletion problem.
The status of PFDS is very similar to that of FVS. It has an approximation preserving reduction from the vertex cover problem and consequently, it is 
NP-hard, and does not have a polynomial time $(2-\epsilon)$-approximation
for every constant $\epsilon >0$ assuming the UGC. It admits a polynomial-time $2$-approximation based on the local-ratio technique \cite{LFFW19}. 
The authors in \cite{LFFW19} do not discuss LP relaxations for PFDS. However, the local-ratio technique of \cite{LFFW19} can be converted to an LP-based $2$-approximation for PFDS following the ideas in \cite{CHUDAK1998111}. 
We describe the associated LP relaxation. A graph $G=(V, E)$ is a \emph{$2$-pseudotree} if it is connected
and has $|E| \geq |V|+1$ (i.e., the graph has at least $2$ edges in
addition to a spanning tree). We will denote the following polyhedra
as weak density polyhedron and $2$-pseudotree cover polyhedron
respectively, and the constraints describing them as weak density
constraints and $2$-pseudotree cover constraints respectively:
\begin{align}
    \weakdensitypolyhedron(G)
    &:=
    \left\{ x\in \R^{V}_{\ge 0}:\ \sum \nolimits_{u\in S}(d_S(u)-1)x_u\geq |E[S]| - |S| \ \forall S \subseteq V
\right\}, \text{ and}  \label{eqn:weak-density}\\
\twopseudotreecoverpolyhedron(G)&:=\left\{x\in \R^V_{\ge 0}: \sum \nolimits_{u\in U}x_u \ge 1 \ \forall U\subseteq V \text{ such that $G[U]$ contains a $2$-pseudotree}\right\}. \label{eqn:2PT-cover}
\end{align}
We encourage the reader to compare and contrast the weak density constraints with the strong density constraints.
Weak density constraints and $2$-pseudotree covering constraints are valid for PFDS. In particular, the following are ILP formulations for PFDS:
\begin{align}
    \min &\left\{\sum \nolimits_{u\in V}c_u x_u:\ x\in \weakdensitypolyhedron(G)\cap \Z^V\right\} \text{ and}\tag{PFDS-IP: WD} \label{PFDS-IP: WD}\\
    \min &\left\{\sum \nolimits_{u\in V}c_u x_u:\ x\in \weakdensitypolyhedron(G)\cap \twopseudotreecoverpolyhedron(G) \cap \Z^V\right\}. \tag{PFDS-IP: WD-and-2PT-cover} \label{PFDS-IP: WD-and-2PT-cover}
\end{align}
The local-ratio technique for PFDS due to Lin, Feng, Fu, and Wang \cite{LFFW19} can be converted to a $2$-approximation for PFDS with respect to the following LP-relaxation via the primal-dual technique:
\begin{align}
    \min\left\{\sum \nolimits_{u\in V}c_u x_u:\ x\in \weakdensitypolyhedron(G)\cap \twopseudotreecoverpolyhedron(G)\right\} \tag{PFDS-LP: WD-and-2PT-cover}. \label{PFDS-LP: WD-and-2PT-cover}
\end{align}
The family of $2$-pseudotree cover constraints admits a polynomial-time separation oracle (see
\Cref{thm:MC2PT-polytime:main}). However, we do not know a
polynomial-time separation oracle for the family of weak density
constraints (similar to the status of strong density constraints), and for this reason we do not know how to solve (\ref{PFDS-LP: WD-and-2PT-cover}) in polynomial time\footnote{We note that Fujito's framework \cite{Fuj-matroid-FVS} encompasses PFDS and hence a $2$-approximation follows from his work.
It appears that \cite{LFFW19} were unaware of \cite{Fuj-matroid-FVS}. The authors of this paper also became aware of \cite{Fuj-matroid-FVS} after completing
an earlier version. The LP relaxation in \cite{Fuj-matroid-FVS} does not appear to have a polynomial-time separation oracle.}.
This leads to the following question (which is the counterpart of Question \ref{q:intro} for PFDS):
\begin{question}
  \label{q:intro-pfds}
  Does there exist an ILP formulation for PFDS whose LP-relaxation can be solved in polynomial time and has integrality gap at most $2$?
\end{question}

We emphasize that weak density constraints and our answer to Question \ref{q:intro-pfds} will play a crucial role in our answer to Question \ref{q:intro}. 
Furthermore, we prove an extreme point property of
the weak density polyhedron that closely resembles the extreme point
property of the strong density polyhedron mentioned in
\Cref{conj:strong-density-extreme-point}.

\subsection{Results}
We begin with our results for PFDS. These results influence our results for FVS which will be discussed subsequently. 
\paragraph{PFDS: ILP formulations and LP integrality gap.} 
We answer Question \ref{q:intro-pfds} affirmatively via a new ILP for PFDS. 
Our ILP for PFDS is based on Charikar's LP for the densest
subgraph problem (DSG) \cite{charikar_greedy_2000}. In DSG, the input is
an undirected graph $G=(V,E)$ and the goal is to find an induced subgraph $G[S]$ of
maximum density where the density of $S \subseteq V$ is defined
as ${|E(S)|}/{|S|}$. We define the density of $G$ to be $\max\{|E[S]|/|S|: \emptyset \ne S\subseteq V\}$. We note that a graph is a pseudoforest if and only if
its density is at most one. Thus, PFDS can equivalently be phrased as
the problem of finding a minimum cost subset of vertices to delete so
that the remaining graph has density at most one.  Charikar formulated
an LP to compute the density of a graph. The dual of Charikar's LP can
be interpreted as a fractional orientation problem.  Using that dual,
we obtain an ILP for PFDS. We describe the details now. Via Charikar's LP
and previous results, one can show that an unweighted graph $G$ has density at most $\lambda$
iff the edges of $G$ can be fractionally oriented such that the total fractional in-degree at
every vertex is at most $\lambda$. For an edge $e=uv$ we use variables $y_{e,u}$
and $y_{e,v}$ to denote the fractional amount of $e$ that is oriented towards $u$ and $v$
respectively. We recall that in PFDS the goal is to remove vertices such that
the residual graph has density at most $1$. Thus, we also have variables
$x_u$ for each $u \in V$ to indicate whether $u$ is deleted. An edge $e=uv$ is in
the residual graph only if $u$ and $v$ are not deleted. These observations
allow us to formulate the an ILP for PFDS based on the polyhedron below. 
We refer to this as the orientation polyhedron:

\begin{align*}
\orientationpolyhedron(G):=
\left\{ (x,y): \begin{array}{l}
x_u + x_v+ y_{e,u} + y_{e,v} \ge 1\ \forall e\in E: e=uv\\
x_u + \sum_{e\in \delta(u)} y_{e,u} \le 1\ \forall u \in V\\
x_u \ge 0\ \forall u\in V\\
y_{e,u} \ge 0\ \forall e\in \delta(u), u\in V.
  \end{array}\right\}.
\end{align*}

We will denote the projection of $\orientationpolyhedron(G)$ to the $x$ variables by $\projectedorientationpolyhedron(G)$. For non-negative costs $c: V\rightarrow \R_{\ge 0}$, we consider the following formulations: 
\begin{align}
\min &\left\{\sum \nolimits_{u\in V}c_u x_u: x\in \projectedorientationpolyhedron(G)\cap \Z^V\right\} \text{ and} \tag{PFDS-IP: orient} \label{PFDS-IP:orient}\\
\min &\left\{\sum \nolimits_{u\in V}c_u x_u: x\in \projectedorientationpolyhedron(G)\cap \twopseudotreecoverpolyhedron(G)\cap \Z^V\right\}. \tag{PFDS-IP: orient-and-2PT-cover} \label{PFDS-IP:Orient-and-2PT-cover} 
\end{align}
We will be interested in the integrality gap of the following LP-relaxation of \eqref{PFDS-IP:Orient-and-2PT-cover}:
\begin{align}
    \min &\left\{\sum \nolimits_{u\in V}c_u x_u: x\in \projectedorientationpolyhedron(G)\cap \twopseudotreecoverpolyhedron(G) \right\}. \tag{PFDS-LP: orient-and-2PT-cover} \label{PFDS-LP:Orient-and-2PT-cover} 
    \end{align}

\begin{restatable}{theorem}{lemmaOrientationFormulation}
\label{lemma:orientation-formulation}
For an input graph $G=(V, E)$ with non-negative costs $c: V\rightarrow \R_{\ge 0}$, 
\eqref{PFDS-IP:orient} and \eqref{PFDS-IP:Orient-and-2PT-cover}  are integer linear programming formulations for PFDS. 
Moreover, we have the following properties:
\begin{enumerate}[label=(\arabic*)]
\item $\projectedorientationpolyhedron(G)\subseteq \weakdensitypolyhedron(G)$ for every graph $G$ and there exist graphs $G$ for which $\projectedorientationpolyhedron(G)\subsetneq \weakdensitypolyhedron(G)$. 

\item 
    The LP \eqref{PFDS-LP:Orient-and-2PT-cover} is solvable in polynomial time and its integrality gap is at most $2$. 
\end{enumerate}
\end{restatable}

We note that a necessary step for showing \Cref{lemma:orientation-formulation}(1) is to show that there exists a polynomial-time separation oracle for the family of $2$-pseudotree cover constraints. We include a proof of this fact in \Cref{thm:MC2PT-polytime:main} in \Cref{sec:2pt-cover-constraints-separation-oracle} for the sake of completeness. We also note that the upper bound of $2$ on the integrality gap of
(\ref{PFDS-LP:Orient-and-2PT-cover}) mentioned in
\Cref{lemma:orientation-formulation} is tight as shown in Figure \ref{fig:orientation-2PT-cover-integrality-gap-instance}.
As a consequence of \Cref{lemma:orientation-formulation}, we
immediately obtain an integer linear program for PFDS
whose LP-relaxation is solvable in polynomial time and has integrality gap at most $2$ --- namely
(\ref{PFDS-IP:Orient-and-2PT-cover}). We also bound the integrality
gap of the LP-relaxation of (\ref{PFDS-IP:orient}), but this will be
based on extreme point properties and will be discussed next.

\begin{figure}[H]
    \centering
\includegraphics[width=0.4\textwidth]{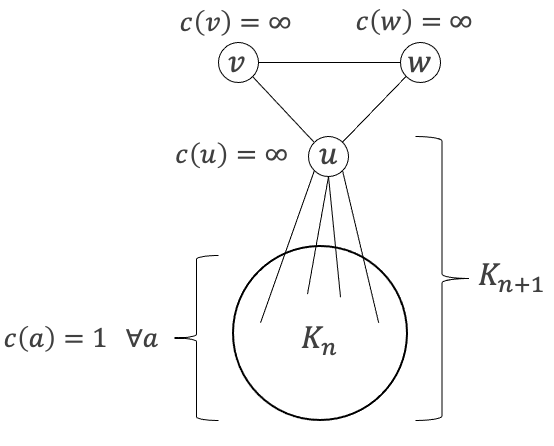}
    \caption{An example showing that the integrality gap of (\ref{PFDS-LP:Orient-and-2PT-cover}) tends to $2$ as $n$ tends to infinity. Consider the graph $G=(V, E)$ as shown above (where $K_n$ denotes the complete graph on $n$ vertices) with cost of every vertex $a\in V-\{u, v, w\}$ being $1$ and costs of vertices $u, v, w$ being infinite. The optimum value of (\ref{PFDS-IP:Orient-and-2PT-cover}) is $n-1$. The optimum value of (\ref{PFDS-LP:Orient-and-2PT-cover}) is at most $n/2$: the solution $x_a = 1/2$ for every vertex $a\in V-\{u, v, w\}$, $x_u = x_v = x_w = 0$, $y_{e, a}=0$ for all edges $e=ab$ where $a, b\in V-\{u, v, w\}$, $y_{uv, v}=y_{vw, w}=y_{wu,u}=1$, $y_{uv, u} = y_{vw, v} = y_{wu, w} = 0$, and $y_{ua,a}=1/2$, $y_{ua, u} = 0$ for every $a\in V-\{u, v,w\}$ is feasible for (\ref{PFDS-LP:Orient-and-2PT-cover}) and has cost $n/2$.}
    \label{fig:orientation-2PT-cover-integrality-gap-instance}
\end{figure}

\paragraph{PFDS: extreme point properties.}
Motivated by \Cref{conj:strong-density-extreme-point} and the goal of
designing primal rounding algorithms for PFDS (and thus, for FVS), 
we investigate extreme point properties of the weak density polyhedron
and the orientation polyhedron. Although we were unable to resolve
\Cref{conj:strong-density-extreme-point} for the strong density
polyhedron, we were able to prove an extreme point property for the
weak density polyhedron.

\begin{restatable}{theorem}{thmWeakDensityExtremePoint}
\label{thm:weak-density-extreme-point}
Let $G=(V, E)$ be a graph that is not a pseudoforest. For every extreme point $x$ of the polyhedron $\weakdensitypolyhedron(G)$, there exists a vertex $u \in V$ such that $x_u\geq 1/3$.
\end{restatable}
Our proof of \Cref{thm:weak-density-extreme-point} is based on a
conditional supermodularity property---if all coordinates are small,
then the weak density constraints have a supermodular property; we use this
supermodular property to show the existence of a
structured basis for the extreme point which is subsequently used to arrive at a contradiction. To the best of authors' knowledge, the conditional
supermodularity property based proof has not previously appeared in the
literature on iterated rounding, and might be of independent interest.
\Cref{thm:weak-density-extreme-point} implies the 
corollary below regarding the following LP-relaxation of (\ref{PFDS-IP: WD}): 
\begin{align}
    \min\left\{\sum \nolimits_{u\in V}c_u x_u: x\in \weakdensitypolyhedron(G)\right\}. \tag{PFDS-LP: WD} \label{PFDS-LP: WD}
\end{align}

\begin{thmcorollary}\label{cor:thm:integrality-gap-3-WD-PFD}
The integrality gap of \eqref{PFDS-LP: WD} is at most $3$. 
\end{thmcorollary}
\Cref{cor:thm:integrality-gap-3-WD-PFD} can be seen to follow from \Cref{thm:weak-density-extreme-point} by the iterative rounding technique, where we repeatedly apply the following two steps until the graph $G$ is a pseudoforest: 
(1) Compute an extreme point optimum solution $x$ of (\ref{PFDS-LP: WD}). 
(2) By \Cref{thm:weak-density-extreme-point}, there exists a vertex $u\in V$ such that $x_u \geq 1/3$; include the vertex $u$ in the solution and remove it from the graph $G$. The approximation factor of the solution constructed by this procedure relative to the starting extreme point optimum solution to the LP is at most $3$. The example given in \Cref{fig:butterfly-WD-integrality-gap} shows that the $(1/3)$-bound in \Cref{thm:weak-density-extreme-point} and  the integrality gap of $3$ mentioned in \Cref{cor:thm:integrality-gap-3-WD-PFD} are both tight.

\begin{figure}[]
\centering
\begin{subfigure}[b]{0.47\textwidth}
    \centering
    \begin{tikzpicture}
        \tikzstyle{vertex}=[circle, draw]
        \node[vertex, label=above:$1/3$](v) at (0, 0) {};
        \node[vertex, label=above:$0$](a) at (-1, 1) {};
        \node[vertex, label=above:$0$](b) at (1, 1) {};
        \node[vertex, label=below:$0$](c) at (-1,-1) {};
        \node[vertex, label=below:$0$](d) at (1,-1) {};
            \begin{scope}[every path/.style={-}, every node/.style={inner sep=1pt}]
                   \draw (a) -- (c);
                   \draw (b) -- (d);
                   \draw (a) -- (v);
                   \draw (b) -- (v);
                   \draw (c) -- (v);
                   \draw (d) -- (v);
            \end{scope} 
        \end{tikzpicture}
        \caption{The unique extreme point in the weak density polyhedron along the all-ones objective direction for the butterfly graph is as shown. The largest coordinate is $1/3$ and all other coordinates are $0$. Consequently, the integrality gap of the LP $\min\{\sum_{u\in V}x_u: x\in \weakdensitypolyhedron(G)\}$ is at least $3$.}
        \label{fig:butterfly-WD-integrality-gap}
     \end{subfigure}
     \hfill
     \begin{subfigure}[b]{0.47\textwidth}
    \centering
    \begin{tikzpicture}
    \small
        \tikzstyle{vertex}=[circle, draw]
        \node[vertex, label=above:$1/3$](v) at (0, 0) {};
        \node[vertex, label=above:$0$](a) at (-2.5, 2.5) {};
        \node[vertex, label=above:$0$](b) at (2.5, 2.5) {};
        \node[vertex, label=below:$0$](c) at (-2.5,-2.5) {};
        \node[vertex, label=below:$0$](d) at (2.5,-2.5) {};
                    \draw (a) -- (c);
                   \draw (b) -- (d);
                   \draw (a) -- (v);
                   \draw (b) -- (v);
                   \draw (c) -- (v);
                   \draw (d) -- (v);

        \draw[->, gray, dashed] (-2.2, -0.5) to (-2.2, 0.5)  node[above] {\footnotesize $2/3$};
        \draw[->, gray, dashed] (-2.8, 0.5) to (-2.8, -0.5) node[below] {\footnotesize $1/3$};

        \draw[->, gray, dashed] (2.2, -0.5) to (2.2, 0.5) node[above] {\footnotesize $2/3$};
        \draw[->, gray, dashed] (2.8, 0.5) to (2.8, -0.5) node[below] {\footnotesize $1/3$};

        \draw[->, gray, dashed] (-1.5, -1.2) to (-0.8, -0.5) node[above] {\footnotesize $0$};
        \draw[->, gray, dashed] (-0.5, -0.8) to (-1.2, -1.5) node[below] {\footnotesize $2/3$};

        \draw[->, gray, dashed] (1.5, 1.2) to (0.8, 0.5) node[below] {\footnotesize $1/3$};
        \draw[->, gray, dashed] (0.5, 0.8) to (1.2, 1.5) node[above] {\footnotesize $1/3$};

        \draw[->, gray, dashed] (-1.5, 1.2) to (-0.8, 0.5) node[below] {\footnotesize $1/3$};
        \draw[->, gray, dashed] (-0.5, 0.8) to (-1.2, 1.5) node[above] {\footnotesize $1/3$};

        \draw[->, gray, dashed] (1.5, -1.2) to (0.8, -0.5) node[above] {\footnotesize $0$};
        \draw[->, gray, dashed] (0.5, -0.8) to (1.2, -1.5) node[below] {\footnotesize $2/3$};
        
        \end{tikzpicture}
        \caption{A minimal extreme point in the orientation polyhedron along the all-ones objective direction for the butterfly graph is as shown. The largest $x$-coordinate is $1/3$ and all other $x$-coordinates are $0$. Consequently, the integrality gap of the LP $\min\{\sum_{u\in V}x_u: x\in \orientationpolyhedron(G)\}$ is at least $3$.}
        \label{fig:butterfly-orientation-integrality-gap}
\end{subfigure}
     
\end{figure}

We emphasize that the extreme point result given in \Cref{thm:weak-density-extreme-point} and the integrality gap result mentioned in Corollary \ref{cor:thm:integrality-gap-3-WD-PFD} are interesting only from a polyhedral viewpoint currently, and are not of help from the perspective of algorithm design. This is because, implementing the above-mentioned iterative rounding procedure requires us to solve the LP-relaxation (\ref{PFDS-LP: WD}) in polynomial time, but we do not know how to do this yet. 

We next show that although $\projectedorientationpolyhedron(G)$ is a subset of $\weakdensitypolyhedron(G)$ (as shown in \Cref{lemma:orientation-formulation}), the extreme point result for $\weakdensitypolyhedron(G)$ given in \Cref{thm:weak-density-extreme-point} still holds for $\projectedorientationpolyhedron(G)$. We will say that an extreme point of a polyhedron is \emph{minimal} if for each variable, reducing the value of that variable by any $\epsilon>0$ while keeping the rest of the variables unchanged results in a point that is outside the polyhedron. We note that extreme points of a polyhedron along non-negative objective directions will be minimal. 

\begin{restatable}{theorem}{thmOrientationPolyhedronExtremePoint}
\label{thm:orientation-polyhedron-extreme-point}
Let $G=(V, E)$ be a graph that is not a pseudoforest. For every minimal extreme point $(x,y)$ of the polyhedron $\orientationpolyhedron(G)$, there exists a vertex $u\in V$ such that $x_u\ge 1/3$. 
\end{restatable}

Similar to \Cref{cor:thm:integrality-gap-3-WD-PFD} that follows from \Cref{thm:weak-density-extreme-point}, \Cref{thm:orientation-polyhedron-extreme-point} implies the corollary below regarding the following LP-relaxation of \eqref{PFDS-IP:orient}: 
\begin{align}
    \min\left\{\sum \nolimits_{u\in V}c_u x_u: x\in \projectedorientationpolyhedron(G)\right\} \tag{PFDS-LP: orient} \label{PFDS-LP:orient}
\end{align}

\begin{thmcorollary}\label{cor:thm:integrality-gap-3-orientation}
The integrality gap of (\ref{PFDS-LP:orient}) is at most $3$. 
Moreover, given an extreme point optimum solution $x$ for the LP, there exists a polynomial time algorithm to obtain an integral feasible solution $x'$ for the LP-relaxation such that $\sum_{u\in V}c_u x'_u \le 3 \sum_{u\in V}c_u x_u$. 
\end{thmcorollary}

\Cref{cor:thm:integrality-gap-3-orientation} can be seen to follow from \Cref{thm:orientation-polyhedron-extreme-point} by the iterative rounding technique, where we repeatedly apply the following two steps until the graph $G$ is a pseudoforest: (1) Compute an extreme point optimum solution $x$ of $\min\left\{\sum_{u\in V}c_u x_u: x\in \projectedorientationpolyhedron(G)\right\}$. (2) By \Cref{thm:orientation-polyhedron-extreme-point}, there exists a vertex $u\in V$ such that $x_u\ge 1/3$; include the vertex $u$ in the solution and remove it from the graph $G$. The approximation factor of this procedure is $3$. In contrast to the weak density constraints based LP, namely (\ref{PFDS-LP: WD}), we can solve the orientation based LP, namely (\ref{PFDS-LP:orient}), in polynomial time. The example given in \Cref{fig:butterfly-orientation-integrality-gap} shows that the $(1/3)$-factor in \Cref{thm:orientation-polyhedron-extreme-point} and the integrality gap of $3$ mentioned in \Cref{cor:thm:integrality-gap-3-orientation} are both tight.

\paragraph{FVS: poly-sized ILP formulations and LP integrality gaps.}
Next, we answer Question \ref{q:intro} affirmatively via a new ILP for FVS. 
In order to achieve this goal, we formulate an intermediate ILP for FVS whose
LP-relaxation has integrality gap at most $2$, but it is unclear if
this LP-relaxation is polynomial-time solvable.  We will later
formulate a \emph{polynomial-sized} ILP for FVS whose LP-relaxation is
at least as strong as that of the intermediate ILP, thereby achieving
our goal. Our intermediate ILP is based on weak density constraints and the following polyhedron that will be referred to as the cycle cover polyhedron: 
\begin{align}
\cyclecoverpolyhedron(G)&:=\left\{x\in \R^V_{\ge 0}: \sum \nolimits_{u\in U}x_u \ge 1 \ \forall U\subseteq V \text{ such that $G[U]$ contains a cycle}\right\}. \label{eqn:cycle-cover}
\end{align}
We will denote the constraints describing the cycle cover polyhedron as cycle cover constraints. 
It is known that the integrality gap of this polyhedron for FVS is $\Theta(\log n)$ \cite{BarYehudaGNR98}.

For non-negative costs $c: V\rightarrow \R_{\ge 0}$, we consider the following formulation and its LP-relaxation:
\begin{align}
\min & \left\{\sum \nolimits_{u\in V}c_u x_u: x\in \weakdensitypolyhedron(G)\cap \cyclecoverpolyhedron(G)\cap \Z^V\right\} \quad \text{and} \tag{FVS-IP: WD-and-cycle-cover} \label{FVS-IP: WD-and-cycle-cover}\\
\min & \left\{\sum \nolimits_{u\in V} c_u x_u: x\in \weakdensitypolyhedron(G)\cap \cyclecoverpolyhedron(G)\right\}. \tag{FVS-LP: WD-and-cycle-cover}
\label{FVS-LP:weak-density-cycle-cover}
\end{align}

\begin{restatable}{theorem}{thmWeakDensityCycleCover}
\label{thm:integrality-gap-wd+cycle-cover}
For an input graph $G=(V, E)$ with non-negative costs $c: V\rightarrow \R_{\ge 0}$, \eqref{FVS-IP: WD-and-cycle-cover} is an integer linear programming formulation for FVS.  
Moreover, the integrality gap of \eqref{FVS-LP:weak-density-cycle-cover} is at most $2$. 
\end{restatable}
\Cref{thm:integrality-gap-wd+cycle-cover} shows that (\ref{FVS-IP: WD-and-cycle-cover}) is an ILP for FVS whose LP-relaxation (\ref{FVS-LP:weak-density-cycle-cover}) has integrality gap at most $2$. However, we do not have a polynomial-time algorithm to solve the LP-relaxation (\ref{FVS-LP:weak-density-cycle-cover}). This is because, we do not have a polynomial-time separation oracle for the family of weak density constraints (although we do have a polynomial-time separation oracle for the family of cycle cover constraints). 

\begin{remark}
In the spirit of using iterative rounding to bound integrality gaps (e.g., \Cref{thm:weak-density-extreme-point} that leads to \Cref{cor:thm:integrality-gap-3-WD-PFD} and \Cref{thm:orientation-polyhedron-extreme-point} that leads to \Cref{cor:thm:integrality-gap-3-orientation}), it is tempting to bound the integrality gap of (\ref{FVS-LP:weak-density-cycle-cover}) by proving an extreme point result. In particular, if we could show that every extreme point optimum of (\ref{FVS-LP:weak-density-cycle-cover}) had a coordinate with value at least $1/2$, then we would have an alternative proof of the integrality gap bound mentioned in \Cref{thm:integrality-gap-wd+cycle-cover} via the  iterative rounding technique. However, the example given in \Cref{fig:K4-WD-CC-integrality-gap} shows that there exists an extreme point optimum of (\ref{FVS-LP:weak-density-cycle-cover}) all of whose coordinates have value at most $1/3$. 
\end{remark}

\begin{figure}[]

     \centering
    \begin{tikzpicture}
        \tikzstyle{vertex}=[circle, draw]
        \node[vertex, label=above:$1/3$](a) at (-1, 1) {};
        \node[vertex, label=above:$1/3$](b) at (1, 1) {};
        \node[vertex, label=below:$1/3$](c) at (-1,-1) {};
        \node[vertex, label=below:$1/3$](d) at (1,-1) {};
            \begin{scope}[every path/.style={-}, every node/.style={inner sep=1pt}]
                   \draw (a) -- (b);
                   \draw (a) -- (c);
                   \draw (a) -- (d);
                   \draw (b) -- (c);
                   \draw (b) -- (d);
                   \draw (c) -- (d);
            \end{scope} 
        \end{tikzpicture}
        \caption{The unique extreme point optimum for (\ref{FVS-LP:weak-density-cycle-cover}) along the all-ones objective direction on the complete graph $K_4$ is as shown. Consequently, the integrality gap of (\ref{FVS-LP:weak-density-cycle-cover}) is at least $3$.}
        \label{fig:K4-WD-CC-integrality-gap}
\end{figure}

We will use \Cref{lemma:orientation-formulation} in conjunction with \Cref{thm:integrality-gap-wd+cycle-cover} to prove the following result: 
\begin{theorem}\label{thm:poly-sized-LP-with-integrality-gap-atmost-2-for-FVS} 
There exists a polynomial-sized integer linear programming formulation for FVS whose LP-relaxation has integrality gap at most $2$. 
\end{theorem}
We prove \Cref{thm:poly-sized-LP-with-integrality-gap-atmost-2-for-FVS} by showing three different polynomial-sized integer linear programs for FVS all of whose LP-relaxations have integrality gap at most $2$:
\begin{enumerate}
\item The first formulation is 
\begin{align*}
    \min\left\{\sum \nolimits_{u\in V}c_u x_u: x\in \projectedorientationpolyhedron(G)\cap \cyclecoverpolyhedron(G)\cap \Z^V\right\}. \tag{FVS-IP: orient-and-cycle-cover} 
\end{align*}
By \Cref{lemma:orientation-formulation}(1), we have that  $\projectedorientationpolyhedron(G)\subseteq \weakdensitypolyhedron(G)$. As a consequence of \Cref{thm:integrality-gap-wd+cycle-cover}, the integrality gap of the following LP-relaxation is also at most $2$: 
\begin{align}
    \min\left\{\sum \nolimits_{u\in V}c_u x_u: x\in \projectedorientationpolyhedron(G)\cap \cyclecoverpolyhedron(G)\right\}.  \tag{FVS-LP: orient-and-cycle-cover} \label{FVS-LP: orient-and-cycle-cover}
\end{align}
Both $\projectedorientationpolyhedron(G)$ and $\cyclecoverpolyhedron(G)$ admit a polynomial-sized description (see \Cref{lemma:cycle-cover-poly-sized-lp} in \Cref{sec:cycle-cover-poly-sized-formulation} for polynomial-sized description of $\cyclecoverpolyhedron(G)$), and thus, we have a polynomial-sized ILP for FVS whose LP-relaxation has integrality gap at most $2$.

\item The second formulation is the Chekuri-Madan formulation who, as
  we remarked earlier, formulated an ILP for Subset-FVS and showed
  that the integrality gap of its LP-relaxation is at most $13$
  \cite{chekuri-madan16}. We show that
  their LP-relaxation specialized for FVS has integrality gap at most
  $2$ by proving that it is at least as strong as
  (\ref{FVS-LP: orient-and-cycle-cover}). Our result gives additional impetus
  to improving the integrality gap of their LP-relaxation for Subset-FVS. 

\item Our third formulation to prove \Cref{thm:poly-sized-LP-with-integrality-gap-atmost-2-for-FVS} is based on the orientation perspective, but without cycle cover constraints (as opposed to our first formulation). Here, we give an orientation based ILP formulation whose associated polyhedron is contained in the strong density polyhedron. Since the integrality gap of (\ref{FVS-LP: SD}) is at most $2$ (as shown by Chudak, Goemans, Hochbaum, and Williamson \cite{CHUDAK1998111}), the integrality gap of our third formulation is also at most $2$.   
\end{enumerate}

We emphasize that the proof of the integrality gap of the
LP-relaxation of all three of our ILPs rely on the integrality gap
upper bound of (\ref{FVS-LP:weak-density-cycle-cover}) or (\ref{FVS-LP:
  SD}). It would be interesting to have a direct proof. In particular,
it would be useful to design a primal rounding algorithm (for arbitrary
LP-optimum solutions or for extreme point LP-optimum solutions) that
yields a $2$-approximate solution.


\paragraph{Organization.} 
We focus on PFDS in Section 3 and FVS in Section 4. In Section \ref{sec:orientatin-polytope}, we prove properties of the orientation polyhedron and show \Cref{lemma:orientation-formulation}. 
In Sections \ref{sec:extreme-point-weak-density} and \ref{sec:orientation-polytope-extreme-point}, we prove extreme point properties of the weak density polyhedron and the orientation polyhedron and show Theorems \ref{thm:weak-density-extreme-point} and \ref{thm:orientation-polyhedron-extreme-point} respectively. 
In Section \ref{sec:integrality-gap-of-weak-density-and-cycle-cover}, we prove properties of the weak density polyhedron and 
show \Cref{thm:integrality-gap-wd+cycle-cover}. 
In Section \ref{sec:orientation-and-cycle-cover-formulation}, we give a polynomial-sized ILP formulation for FVS with integrality gap at most $2$, thereby proving \Cref{thm:poly-sized-LP-with-integrality-gap-atmost-2-for-FVS}. For the sake of completeness, we discuss two additional polynomial-sized ILP formulations for FVS with integrality gap at most $2$ in Section \ref{sec:more-FVS-formulations}. 

\section{Preliminaries}
For a minimization IP $\min\{c^Tx: Ax\le b, x\in \Z^n_{\geq 0}\}$, the \emph{integrality gap} of its LP-relaxation $\min\{c^Tx: Ax\le b, x\geq 0\}$ is the maximum ratio, over all non-negative cost functions, of the minimum cost of a feasible solution to the IP and that of the LP, i.e., $\max_{c\in \R^n} \min\{c^Tx: Ax\le b, x\in \Z^n_{\geq 0}\}/\min\{c^Tx: Ax\le b, x \geq 0\}$.
For a graph $G = (V, E)$, we will use $\delta(A, B)$ to denote the set of edges with exactly one end-vertex in $A$ and exactly one end-vertex in $B$.
We summarize the results of Lin, Feng, Fu, and Wang \cite{LFFW19} in a manner that will be useful for our polyhedral study.  
\begin{theorem}\label{thm:PFDS-WD-and-2PT-cover}
(\ref{PFDS-IP: WD}) and (\ref{PFDS-IP: WD-and-2PT-cover}) are ILP formulations for PFDS. Furthemore, the LP-relaxation (\ref{PFDS-LP: WD-and-2PT-cover}) has integrality gap at most $2$. 
\end{theorem}
We emphasize that Lin, Feng, Fu, and Wang did not prove \Cref{thm:PFDS-WD-and-2PT-cover} directly, but their local ratio based algorithm and standard techniques to convert local ratio based algorithms to LP-based primal-dual algorithms (following the ideas in \cite{CHUDAK1998111}) lead to the theorem.  

\section{Pseudoforest Deletion Set}
\label{sec:PFDS}
In this section, we formally prove our main results for Pseudoforest Deletion Set. We prove \Cref{lemma:orientation-formulation} in \Cref{sec:orientatin-polytope}, \Cref{thm:weak-density-extreme-point} in \Cref{sec:extreme-point-weak-density}, and \Cref{thm:orientation-polyhedron-extreme-point} in \Cref{sec:orientation-polytope-extreme-point}.

\subsection{Orientation and $2$-Pseudotree Cover constraints}\label{sec:orientatin-polytope}
In this section, we restate and prove \Cref{lemma:orientation-formulation}. 

\lemmaOrientationFormulation*
\begin{proof}
In order to show that (\ref{PFDS-IP:orient}) and (\ref{PFDS-IP:Orient-and-2PT-cover}) are ILP formulations for PFDS, it suffices to show that the ILP (\ref{PFDS-IP:orient}) is a formulation for PFDS. This is because $2$-pseudotree cover constraints are valid for PFDS.

For a subgraph $\Tilde{G} = (\Tilde{V}, \Tilde{E})$ of the graph $G$, we let $d_{\Tilde{G}}(u):=|\delta(u)\cap \Tilde{E}|$ denote the degree of the vertex $u$ in the subgraph $\Tilde{G}$. We now define an intermediate polyhedron $\weakdensitypolyhedronforsubgraphs(G)$.
\begin{align}
\weakdensitypolyhedronforsubgraphs(G) 
    &:= 
    \left\{ x\in \R^{V}_{\ge 0}:\ \sum_{u\in \Tilde{V}}(d_{\Tilde{G}}(u)-1)x_u\geq |\Tilde{E}| - |\Tilde{V}| \ \forall \text{ subgraphs $\Tilde{G} = (\Tilde{V}, \Tilde{E})$ of $G$}.
\right\}. \label{eqn:weak-density-for-subgraphs}
\end{align}
We note that $\weakdensitypolyhedronforsubgraphs(G) \subseteq \weakdensitypolyhedron(G)$. We will use the following claim to show that the ILP  (\ref{PFDS-IP:orient}) formulates PFDS and also to conclude that $\projectedorientationpolyhedron(G)\subseteq \weakdensitypolyhedron(G)$. 

\begin{claim}\label{claim:orientation-inside-weak-density-subgraphs}
$\projectedorientationpolyhedron(G)\subseteq \weakdensitypolyhedronforsubgraphs(G)$.
\end{claim}
\begin{proof}
Let $x\in \projectedorientationpolyhedron(G)$. Then, there exists a vector $y$ such that $(x, y)\in \orientationpolyhedron(G)$. Let $\Tilde{G} = (\Tilde{V}, \Tilde{E})$ be an arbitrary subgraph of $G$. We have the following:
\begin{align*}
|\Tilde{E}| &\leq \sum_{e = vw \in \Tilde{E}} (x_v + x_w + y_{e,v} + y_{e,w})\\
       &= \sum_{v \in \Tilde{V}} (d_{\Tilde{G}}(v)-1) x_v + \sum_{v \in \Tilde{V}} \left( x_v + \sum_{e \in \delta_{\Tilde{G}}(v)} y_{e,v}\right)\\ 
       &\leq \sum_{v \in \Tilde{V}} (d_{\Tilde{G}}(v)-1) x_v + |\Tilde{V}|,
\end{align*}
where the inequalities are because the vector $(x,y) \in \orientationpolyhedron$. Rearranging the above inequality gives the constraint for $\tilde{G}$ given in $\weakdensitypolyhedronforsubgraphs(G)$. Hence, $x\in \weakdensitypolyhedronforsubgraphs(G)$.      
\end{proof}

We remark that \Cref{claim:orientation-inside-weak-density-subgraphs} can be strengthened to show that $\projectedorientationpolyhedron(G)= \weakdensitypolyhedronforsubgraphs(G)$, but we will not need this stronger version for the purposes of this theorem. We now show that the ILP (\ref{PFDS-IP:orient}) formulates PFDS. 
Let $S \subseteq V$ be a pseudoforest deletion set, i.e., 
 the subgraph $F := G - S$ is a pseudoforest. Let $x := \chi^S \in \{0,1\}^{V}$ denote the indicator vector of the set $S$.  
 Select an arbitrary orientation of the pseudoforest $F$ such that the maximum indegree of every vertex is at most $1$, and let $y_{e,v} := 1$ for the edge $e = vw$ if $e$ is an edge of $F$ that is oriented towards $v$, and $y_{e,v} := 0$ otherwise. 
 Then, $(x,y) \in \orientationpolyhedron(G)$. 

Next, suppose $x\in \projectedorientationpolyhedron(G)\cap \Z^V$. Then, we have that $x\in \{0,1\}^V$. By \Cref{claim:orientation-inside-weak-density-subgraphs}, $x\in \weakdensitypolyhedronforsubgraphs(G)\subseteq \weakdensitypolyhedron(G)$. 
Since (\ref{PFDS-IP: WD}) is an ILP formulation for PFDS by \Cref{thm:PFDS-WD-and-2PT-cover}, it follows that the set $S:=\{u\in V: x_u =1\}$ is a pseudoforest deletion set for the graph $G$. This concludes the proof that the ILP (\ref{PFDS-IP:orient}) formulates PFDS. 

We now prove the two additional conclusions of the theorem statement. 

\begin{enumerate}[label=(\arabic*)]
\item By Claim \ref{claim:orientation-inside-weak-density-subgraphs}, we have that $\projectedorientationpolyhedron(G)\subseteq \weakdensitypolyhedronforsubgraphs(G)\subseteq \weakdensitypolyhedron(G)$. 
We now show that there is a graph $G$ such that $\weakdensitypolyhedronforsubgraphs(G)\subsetneq \weakdensitypolyhedron(G)$. In particular, we consider the graph $K_5 = (V, E)$ where $V = \{v_1, v_2, \ldots, v_5\}$ and $E = {V \choose 2}$. Let $x = (2/3, 2/3, 1/3, 0, 0)$. We note that $x \in \weakdensitypolyhedron(K_5)$.
We now show that $x \not \in \weakdensitypolyhedronforsubgraphs(K_5)$.
    Consider the subgraph $\Tilde{G} = (\Tilde{V}, \Tilde{E})$ obtained by removing the edge $\{v_1, v_2\}$ from the graph $G$, i.e. $\Tilde{V} = V$ and $\Tilde{E} = {V\choose 2} - \{v_1, v_2\}$.
    Then, we have the following:
    $$\sum_{u \in \Tilde{V}}(d_{\Tilde{G}}(u) - 1)x_u = 2x_1 + 2x_2 + 3x_3 + 3x_4 + 3x_5 = \frac{11}{3} < 4 = |\Tilde{E}| - |\Tilde{V}|.$$
    
    In particular, the vector $x$ does not satisfy the constraint of $\weakdensitypolyhedronforsubgraphs(K_5)$ defined by the subgraph $\Tilde{G}$. Thus, we have that $\projectedorientationpolyhedron(K_5) \subseteq \weakdensitypolyhedronforsubgraphs(K_5)\subset \weakdensitypolyhedron(K_5)$. 

    \item By \Cref{thm:PFDS-WD-and-2PT-cover}, the ILP (\ref{PFDS-IP: WD-and-2PT-cover}) is a valid formulation of PFDS and its LP-relaxation (\ref{PFDS-LP: WD-and-2PT-cover}) has integrality gap at most $2$. We have already shown that the ILP (\ref{PFDS-IP:Orient-and-2PT-cover}) is a valid formulation of PFDS and that $\projectedorientationpolyhedron(G)\subseteq \weakdensitypolyhedron(G)$. Consequently, the integrality gap of (\ref{PFDS-LP:Orient-and-2PT-cover}) is at most $2$. Finally, we note that (\ref{PFDS-LP: WD-and-2PT-cover}) can be solved in polynomial time because $\projectedorientationpolyhedron(G)$ is the projection of $\orientationpolyhedron(G)$ which admits a polynomial sized description, and because 2-pseudotree-cover constraints admit a polynomial time separation oracle by \Cref{thm:MC2PT-polytime:main} (in \Cref{sec:2pt-cover-constraints-separation-oracle}). The polynomial time separation oracle for 2-pseudotree-cover constraints is based on the fact that node-weighted Steiner tree for constant number of terminals is solvable in polynomial time. We present its proof in a separate subsection for clarity. 

\end{enumerate}
\end{proof}



\subsection{Extreme point property of the Weak Density polyhedron}\label{sec:extreme-point-weak-density}
In this section, we prove \Cref{thm:weak-density-extreme-point}. We restate it below.
\thmWeakDensityExtremePoint*
For a fixed $x\in \R^V$, we define the function $f_x: 2^V\rightarrow \R$ as 
\[
f_x(S) := |E[S]|-|S| - \sum_{u\in S}(d_S(u)-1)x_u\ \forall S\subseteq V. 
\]
Using this function, we can express the weak density polyhedron as 
\[
    \weakdensitypolyhedron(G)
     = \left\{x\in\nnreal^{V} : f_x(S) \leq 0\  \forall S \subseteq V\right\}.
\]

 In \Cref{sec:supermodular-background}, we recall properties of supermodular functions and chain set families. 
In  \Cref{sec:conditional-supermodularity-of-weak-density-constraints}, we show that if $x_u\le 1/2$ for all $u\in V$, then the function $f_x:2^V\rightarrow \R$ is supermodular. This conditional supermodularity property allows us to uncross the tight constraints that form a basis of an extreme point $x$ for which $x_u\le 1/2$ for all $u\in V$. We prove properties about the tight constraints in  \Cref{sec:conditional-properties-tight-sets} and show the existence of a \emph{chain basis} in  \Cref{sec:conditional-basis-structure-for-extreme-point}. Using the chain basis structure and its properties, we complete the proof of  \Cref{thm:weak-density-extreme-point} in \Cref{sec:counting-argument}. 

\paragraph{Notation.} Let $b(S):=|E[S]|-|S|$ for all $S\subseteq V$. For vertex $u\in V$, let $\chi_u \in \{0,1\}^V$ denote the indicator vector of $u$. For a subset $S\subseteq V$, let $\row(S)$ denote the row of the constraint matrix of $\weakdensitypolyhedron(G)$ corresponding to the set $S$, i.e, $$\row(S)_u = \begin{cases}
    d_S(u) - 1 & \text{ if $u \in S$} \\
    0 & \text{ o.w.}\\
\end{cases} \qquad \forall u \in V.$$
For $J \subseteq 2^V$, let $\Rows(J) = \{\row(S):S \in J\}$. When the context is clear, we also use $\Rows(J)$ to denote the submatrix of the weak density polyhedron constraint matrix given by the set $\Rows(J)$.
We let $\spanfunc(\Rows(J))$ denote the \emph{span} of the set of vectors $\Rows(J)$, i.e. the smallest linear subspace that contains the set  $\Rows(J)$.
We let $\mathsf{columns}(J)$ denote the set of columns of the submatrix $\Rows(J)$. For matrix $M$, we denote $\mathsf{rank}(M)$ as the \emph{rank} of the matrix. We note that $\mathsf{rank}(\Rows(J)) = \mathsf{rank}(\mathsf{columns}(J))$. For $x\in \R^V$, we will use $g_x(S):=\sum_{u\in S}(d_S(u)-1)x_u$ for all $S\subseteq V$, $\calT_x:=\{S\subseteq V: f_x(S)=0\}$ be the tight sets for $x$, and $\calZ:=\{u\in V: x_u=0\}$.
 Let $x \in P(G)$ be an extreme point of the weak density polyhedron. For $J \subseteq \calT_x$, we say that the set $\Rows(J)$ is a \emph{basis} for $x$ if $\spanfunc(\Rows(J)) = \spanfunc(\Rows(2^V))$.

\subsubsection{Background on supermodular functions and set families}\label{sec:supermodular-background}
Here, we recall supermodular functions, chain set families, and some of their properties that will be useful while proving the extreme point property of the weak density polyhedron. 
A set function $f:2^V\rightarrow \mathbb{R}$ is said to be \emph{supermodular} if $f(A) + f(B) \leq f(A\cap B) + f(A\cup B)$ for all subsets $A, B \subseteq V$. 
We refer the reader to \cite{Fujishige05} for background on supermodular set functions and their properties. 
We will rely on the following property:
 
\begin{proposition}\label{prop:supermodularity-of-graph-edge-functions}
Let $G = (V, E)$ be an undirected graph with non-negative edge weights $w: E\rightarrow \R_{\ge 0}$. Then, the function $f: 2^V\rightarrow \R_{\ge 0}$ defined by $f(S) := \sum_{e\in E[S]}w_e$ for all $S\subseteq V$ is supermodular. 
\end{proposition}
Two sets $A$ and $B$ are said to \emph{cross} if they have a non-empty intersection and neither set is contained in the other, i.e. $A\cap B, A - B, B - A \not = \emptyset$.
For a ground set $V$, the set $\calC = \{A_1, A_2, \ldots A_\ell \} \subseteq 2^{V}$ is said to be a \emph{chain family} if its elements can be ordered such that $A_1 \subseteq A_2 \subseteq \ldots \subseteq A_\ell$. We need the following proposition on chain families. This is well-known but we give
a proof for the sake of completeness. 

\begin{proposition}\label{prop:chain-uncrossing}
Let $\calC$ be a chain family and $A$ be a subset that crosses some set $B\in \calC$. Then, the number of sets in $\calC$ crossed by $A\cup B$ and $A\cap B$ is strictly less than the number of sets in $\calC$ crossed by $A$.
\end{proposition}
\begin{proof}
Our strategy will be to pick an arbitrary set in $\calC$ that crosses the set $A\cup B$ (resp. $A\cap B$) and show that this picked set also crosses the set $A$. The claim then follows since the set $B$ crosses the set $A$ but not the set $A\cup B$ (resp $A\cap B$).

Let $P\in\calC$ be an arbitrary set that crosses the set $A\cup B$. Since both $B, P\in \calC$, it must be that either $P \subset B$ or $B\subset P$. First, we consider the case where $P \subset B$. Then $P \subset A\cup B$ and thus does not cross the set $A\cup B$ as it is contained in it. This contradicts the choice of $P$. Next, we consider the case where $B\subset P$. If $A \subseteq P$ then $A\cup B \subseteq P$, contradicting choice of $P$ crossing $A\cup B$. Furthermore, if $P \subseteq A$, then $B \subset P \subseteq A$, contradicting the hypothesis that the set $B$ crosses the set $A$. Finally, $A\cap P \not = \emptyset$ as $A\cap B\not = 0$ and $B\subset P$. Thus the set $P$ crosses the set $A$.

Let $P\in\calC$ be an arbitrary set that crosses the set $A\cap B$. Since the sets $P, B \in \calC$ are part of a chain family, it must be that either $P\subseteq B$ or $B\subseteq P$. If $B \subseteq P$, then $A\cap B \subseteq P$, contradicting the choice of set $P$. Thus $P \subseteq B$. If $P \subseteq A$, then $P \subseteq A\cap B$, contradicting the choice of set $P$. Moreover, if $A \subseteq P$, then we have that $A \subseteq B$, contradicting that $A$ and $B$ cross. Thus the set $P$ crosses the set $A$.
\end{proof}

\subsubsection{Conditional supermodularity of Weak Density constraints}\label{sec:conditional-supermodularity-of-weak-density-constraints}
We now show that the function $f_x$ is supermodular if  $x_u < \frac{1}{2}$ $\forall u \in V$. 
\begin{lemma}\label{lem:supermodularity}
Let $x\in \R^V$. If $x_u < \frac{1}{2}\ \forall u\in V$, then the function $f_x$ is  supermodular.
\end{lemma}
\begin{proof}
Let $S\subseteq V$. We have 
\begin{align*}
    f_x(S)
    & = |E[S]| - |S| - \sum_{u\in S}(d_S(u) - 1)x_u&\\
    & = |E[S]| - \sum_{u\in S}d_S(u)x_u - |S| + \sum_{u\in S}x_u&\\
    &= |E[S]| - \sum_{u \in S}\sum_{v \in \delta(u)}x_u - \sum_{u \in S}(1 - x_u)&\\
    & = |E[S]| - \sum_{uv\in E[S]}(x_u + x_v) - \sum_{u \in S}(1 - x_u)&\\
    &= \sum_{uv\in E[S]}\left(1 - (x_u+x_v)\right) - \sum_{u\in S}\left(1 - x_u\right).&
\end{align*}
Thus, the function $f_x$ can be expressed as $f_x=p_x-q_x$ for functions $p_x, q_x: 2^V\rightarrow \R$ defined by $p_x(S) :=\sum_{uv\in E[S]}(1 - (x_u + x_v))$ and $q_x(S) := \sum_{u\in S}(1 - x_u)$.
Since $x_u<1/2$ for all $u\in V$, we have that $p_x(S)\ge 0$ and $q_x(S)\ge 0$ for all $S\subseteq V$. 

Now, let $H=(V, E)$ be a graph with edge weights $w:E\rightarrow \R$ defined by $w(uv):=1-(x_u+x_v)$. 
We note that the edge weights are non-negative since $x_a<1/2$ for all $a\in V$. 
By definition, $p_x(S)=\sum_{uv\in E[S]}w(uv)$ for all $S\subseteq V$. By \Cref{prop:supermodularity-of-graph-edge-functions}, the function $p_x$ is supermodular. Moreover, the function $q_x$ is modular. Thus, the function $f_x$ is the sum of a supermodular function and a modular function. Consequently, $f_x$ is a supermodular function. 
\end{proof}

\subsubsection{Conditional properties of tight sets}\label{sec:conditional-properties-tight-sets}
Next, we prove certain properties of tight sets that help us to prove the existence of a well-structured basis for the extreme point $x$ under the condition that $x_u<1/2$ for every $u\in V$. 

\begin{lemma}[Conditional Uncrossing Properties]\label{lem:tight-sets-conditional-uncrossing-properties}
Let $x$ be an extreme point of $\weakdensitypolyhedron(G)$ such that $x_u < \frac{1}{2}$ for all $u\in V$ and the family of tight sets for $x$ be $\calT_x := \{S\subseteq V : f_x(S) = 0\}$. 
Let $A,B\in\calT$. Then,
\begin{enumerate}
    \item $A\cap B \neq \emptyset$, i.e., tight sets overlap,
    \item $A\cap B, A\cup B\in \calT_x$, i.e. tight sets form a lattice family, 
    \item $\delta(A-B, B-A)=\emptyset$, i.e. tight sets admit no crossing edges, and 
    \item $\row(A) + \row(B) = \row(A\cap B) +\row(A\cup B)$.
\end{enumerate}
\end{lemma}
\begin{proof} 
We recall that for a set $S \subseteq V$, we denote $b(S):=|E[S]|-|S|$. Moreover, for the point $x \in P_{WD}(G)$, we denote $g_x(S):=\sum_{u\in S}(d_S(u)-1)x_u$. We prove each property separately below.
\begin{enumerate}
    \item By way of contradiction, assume that $A\cap B = \emptyset$. We have that
    \begin{align*}
        b(A\cup B)&\leq g_x(A\cup B)&\\
        & = g_x(A) + g_x(B) + \sum_{ab \in \delta(A,B)}(x_a + x_b)&\\
        & = |E[A] - |A| + |E[B] - |B| + \sum_{ab \in \delta(A,B)}(x_a + x_b)&\\
        & = |E[A\cup B]| - |A\cup B|- |\delta(A, B)| + \sum_{ab \in \delta(A,B)}(x_a + x_b)&\\
        & = b(A \cup B) - |\delta(A, B)| + \sum_{ab \in \delta(A,B)}(x_a + x_b)&\\
        &< b(A \cup B) - |\delta(A, B)| + |\delta(A, B)|,&
    \end{align*}
    a contradiction. Here, the first inequality is by the weak density constraint for the set $A\cup B$. The 
    first and third equalities are by the hypothesis that $A\cap B = \emptyset$. The final inequality is because $x_u < \frac{1}{2}$ $\forall u \in V$.
    
    \item We have the following: 
    $$0 = f_x(A) + f_x(B) \leq f_x(A\cup B) + f_x(A\cap B) \leq 0.$$
    Here, the first equality is due to the sets $A, B \in \calT_x$, and thus $f_x(A) = f_x(B) = 0$. The first inequality follows from supermodularity of the function $f_x$ shown in \Cref{lem:supermodularity}. The final inequality follows from weak density constraints for the sets $A\cup B$ and $A\cap B$.
    Thus, all inequalities are  equalities, and we have that $f_x(A\cap B) = f_x(A\cup B) = 0$ using weak density constraints on the respective sets.
    
    \item By way of contradiction, assume that $\delta(A-B, B-A)\not = \emptyset$.
    Since $A,B\in\calT_x$, we have that $A\cap B$, $A\cup B \in\calT_x$ by the previously shown property that tight sets uncross. Thus, $f_x(A) = f_x(B) = f_x(A\cap B) = f_x(A\cup B) = 0$. 
    We have that 
    \begin{align*}
        0  = f_x(A) + f_x(B) - f_x(A\cap B) - f_x(A\cup B)
         = \sum_{uv\in \delta(A-B, B-A)}\left(1 - (x_u+x_v)\right)
        > 0,
    \end{align*}
    a contradiction. Here, the second equality is by definition of $f_x$, and the final inequality is because $x_u <\frac{1}{2}$ $\forall u\in V$.

\item Let $u \in V$ be an arbitrary vertex. It suffices to show that $(d_A(u) - 1)+ (d_B(u) - 1) = (d_{A\cup B}(u) - 1) + (d_{A\cap B}(u) - 1)$. We consider four cases:

First, we consider the case $u \in A-B$. Then, we have that (1) $d_A(u) = d_{A-B}(u)+\delta(u, A\cap B)$, (2) $d_B(u) = 0$ (3) $d_{A\cap B}(u) = 0$ and (4) $d_{A\cup B}(u) = d_{A-B}(u)+\delta(u, A\cap B)$. Here, (4) follows from $\delta(u, B - A) = \emptyset$ by the previous part. Thus, the claimed equality follows. We note that the argument for the case where $u \in B-A$ is similar to this case.

Next, we consider the case $u \in A\cap B$. Then, we have that (1) $d_A(u) = d_{A\cap B}(u)+\delta(u, A - B)$, (2) $d_B(u) = d_{A\cap B}(u)+\delta(u, B - A)$,  and (3) $d_{A\cup B}(u) = d_{A\cap B}(u)+\delta(u, A - B)+ \delta(u, B - A)$. Thus, the claimed equality follows.

Finally, we consider the case $u \in V - (A\cup B)$. Then, we have that $d_A(u) =d_B(u) =d_{A\cap B}(u) =d_{A\cup B}(u) = 0$ and the claimed equality follows.
\end{enumerate}
\end{proof}

Next, we show that every tight set is of size at least $2$ and the graph induced over the tight set is connected. 
\begin{lemma}\label{lem:tight-set-connected}
Let $x$ be an extreme point of $\weakdensitypolyhedron(G)$ such that $x_u < \frac{1}{2}$ for all $u\in V$ and the family of tight sets for $x$ be $\calT_x := \{S\subseteq V : f_x(S) = 0\}$. For every $A\in \calT_x$, we have that $|A|\ge 2$ and the graph $G[A]$ is connected. 
\end{lemma}
\begin{proof}
Let $A\in\calT_x$. 
Suppose that $|A|=1$. Let $A=\{u\}$. Then, $A\in \calT_x$ implies that $f_x(\{u\})=0$. Equivalently, $-x_u=-1$ and hence, $x_u=1$, a contradiction. Hence, $|A|\ge 2$. Next, we show that $G[A]$ is connected. By way of contradiction, let $A = A_1\uplus A_2$ such that $G[A_1], G[A_2]$ are disconnected components of $G[A]$. Then, we have the following:
\begin{align*}
    g_x(A_1) + g_x(A_2)&= g_x(A) &\\
    &= b(A)&\\
    &= |E[A]| - |A|&\\
    &= |E[A_1]| - |A_1| + |E[A_2]| - |A_2|&\\
    &= b(A_1) + b(A_2)&\\
    &\leq g_x(A_1) + g_x(A_2).&
\end{align*}
Here, the first and fourth equalities are because $A_1\cap A_2 = \emptyset$. The second equality holds because $A\in\calT_x$. The final inequality holds by the weak density constraints on $A_1$ and $A_2$.

The chain of inequalities implies that the final inequality is an equality. By weak density constraints for $A_1$ and $A_2$, we also have that $b(A_1) \le g_x(A_1)$ and $b(A_2)\le g_x(A_2)$ and consequently, $A_1$ and $A_2$ are tight sets, i.e., $A_1, A_2\in\calT_x$. However, $A_1\cap A_2 = \emptyset$ by assumption, contradicting \Cref{lem:tight-sets-conditional-uncrossing-properties} that tight sets must overlap.
\end{proof}

\subsubsection{Conditional basis structure for extreme points}\label{sec:conditional-basis-structure-for-extreme-point}
In this section, we use the conditional structural properties of tight sets proved in \Cref{sec:conditional-properties-tight-sets} to show that every extreme point $x$ for the weak density polyhedron for which $x_u<1/2$ for all $u\in V$ has a well-structured basis. 
We recall that a $2$-pseudotree is a connected graph that has one more edge than the number of vertices.
The following lemma is the main result of this section.

\begin{lemma}\label{lem:chain-basis-existence}
Let $x$ be an extreme point of $\weakdensitypolyhedron(G)$ such that $x_u < \frac{1}{2}$ $\forall u\in V$, the family of tight sets for $x$ be $\calT_x := \{S\subseteq V : f_x(S) = 0\}$, and let $\calZ := \left\{u \in V : x_u = 0\right\}$. 
Then, there exists a family $\calC\subseteq\calT_x$ such that
\begin{enumerate}[label=(\arabic*)]
    \item the family $\calC$ is a chain family,
    \item the set of vectors $\Rows(\calC\cup\calZ)$ is linearly independent,
    \item $\spanfunc(\Rows(\calC\cup\calZ)) = \spanfunc(\Rows(\calT_x\cup\calZ))$,
    \item For each $S \in \calC$, the subgraph $G[S]$ contains a $2$-psuedotree, 
    \item For each $S \in \calC$ and each vertex $u \in S$, we have that $d_S(u) \geq 2$, and 
    \item For every $A, B\in \calC$ such that $A\subset B$, there exists a vertex $v\in B-A$ such that $x_v>0$. 
\end{enumerate}
\end{lemma}
\begin{proof}
    We first show that there exists a family satisfying properties (1)-(4). Let $\calC^{(1)} \subseteq \calT_x$ be an inclusion-wise maximal chain family. \Cref{claim:Cspan=Tspan} below shows that $\spanfunc(\Rows(\calC^{(1)})) = \spanfunc(\Rows(\calT_x))$.

\begin{claim}\label{claim:Cspan=Tspan}
$\spanfunc(\Rows(\calC^{(1)})) = \spanfunc(\Rows(\calT_x)).$
\end{claim}
\begin{proof}
 By way of contradiction assume false. Let $A \in \calT_x$ such that $\row(A) \not \in \spanfunc(\Rows(\calC^{(1)}))$ and $A$ crosses the fewest number of sets in $\calC^{(1)}$. Consider a set $B\in \calC^{(1)}$ such that $B$ crosses $A$. We note that such a set $B$ exists since otherwise the family $\calC^{(1)} \cup \{A\}$ contradicts the inclusion-wsie maximality of $\calC^{(1)}$. Recall that by \Cref{lem:tight-sets-conditional-uncrossing-properties}, the sets $A\cap B, A\cup B \in \calT_x$ are also tight. We note that by \Cref{prop:chain-uncrossing}, the sets $A\cap B$ and $A\cup B$ cross fewer sets in $\calC^{(1)}$ than the number of sets in $\calC^{(1)}$ crossed by $A$. We consider two cases based on whether $\row(A\cup B), \row(A\cap B) \in \spanfunc(\Rows(\calC^{(1)}))$. First, consider the case where $\row(A\cup B) \not \in \spanfunc(\Rows(\calC^{(1)}))$ without loss of generality. Since $A\cup B$ crosses fewer sets in $\calC$, the set $A\cup B$ contradicts the choice of $A$. 
Next, consider the case where $\row(A\cup B), \row(A\cap B)  \in \spanfunc(\Rows(\calC^{(1)}))$.
By \Cref{lem:tight-sets-conditional-uncrossing-properties}, 
we have that $\row(A) + \row(B) = \row(A\cup B) + \row(A\cap B)$. Thus $\row(A) \in \spanfunc(\Rows(B, A\cap B, A\cup B))$, contradicting choice of $A$.
\end{proof}

    Let $\calC^{(2)} \subseteq \calC^{(1)}$ be an inclusion-wise maximal family such that the set $\Rows(\calC^{(2)} \cup \calZ)$ is linearly independent. We note that 
    $\spanfunc(\Rows(\calC^{(2)}\cup\calZ)) 
 = \spanfunc(\Rows(\calC^{(1)}\cup\calZ)) = \spanfunc(\Rows(\calT_x\cup\calZ))$, where the first equality is because the family $\calC^{(2)}$ is inclusion-wise maximal. In particular, we have that the family $\calC^{(2)}$ satisfies properties (1)-(3).
\Cref{claim:2-pseudotree-containment} below shows that the family $\calC^{(2)}$ also satisfies property (4). 

\begin{claim}\label{claim:2-pseudotree-containment}
For every $S\in \calC^{(2)}$, the subgraph $G[S]$ contains a $2$-pseudotree.
\end{claim}
\begin{proof}
By way of contradiction, let $S\in\calC^{(2)}$ such that $G[S]$ does not contain a $2$-pseudotree. Let $S^{i}:=\{u\in S: d_S(u)=i\}$ and  $S^{\geq i}:=\{u\in S: d_S(u)\geq i\}$. 
We note that $b(S) \leq 0$ since $G[S]$ does not contain a $2$-pseudotree. Hence, we have that
\begin{align*}
    0 \ \geq \ b(S) \ &= g(S)&\\
    &= \sum_{u\in S}(d_S(u)-1)x_u&\\
    &= \sum_{u\in S^{0}}(d_S(u)-1)x_u + \sum_{u\in S^{1}}(d_S(u)-1)x_u + \sum_{u\in S^{\geq 2}}(d_S(u)-1)x_u&\\
    &\geq \sum_{u\in S^{\geq 2}}x_u&\\
    &\geq 0&
\end{align*}
Here, the first equality holds because $S\in\calT_x$. The final inequality holds since $|S|\ge 2$ and $G[S]$ has no isolated vertices by  \Cref{lem:tight-set-connected}, and by the non-negativity constraints on vertex variables of $\weakdensitypolyhedron(G)$.

Thus, all inequalities above are equalities, and consequently, we have that $\sum_{u\in S^{\geq 2}}x_u = 0$ 
This, coupled with the non-negativity constraints on vertex variables, implies that $x_u = 0$ for each $u\in S^{\geq 2}$. We also note that $\row(S)_u = 0$ for each $u \in S^{ 1}\cup S^{0}$ as $S^{0} = \emptyset$ and $d_S(u) - 1 = 0$ for $u\in S^1$.  Thus, we have that $\row(S) \in \spanfunc(\Rows(\calZ))$, contradicting linear independence of $\Rows(\calC^{(2)}\cup\calZ)$. 
\end{proof}

We now show the existence of a family satisfying properties (1)-(5). Let $\calC\subseteq\calT_x$ be a family satisfying properties (1)-(4) which minimizes $\sum_{S \in\calC}|S|$. We will show that this $\calC$ also satisfies property (5).
    By way of contradiction, suppose there exists $S\in\calC$ with $u \in S$ such that $d_S(u) < 2$. Since $S\in\calT_x$, we have that $|S|\ge 2$ and $d_S(u)\geq 1$ since tight sets cannot have isolated vertices by \Cref{lem:tight-set-connected}. Thus, $d_S(u) = 1$. Let $uv \in E[S]$ be the unique edge incident to $u$ in $G[S]$. 
We note that $E[S]$ should contain at least one more edge apart from the edge $uv$ as otherwise $\row(S) = 0$. Furthermore, by \Cref{lem:tight-set-connected}, the subgraph $G[S]$ must be connected. Thus, $d_S(v) \geq 2$ and we have the following: 
\begin{align*}
    g_x(S - u)& = g_x(S) - x_v&\\
    & = b(S) - x_v&\\
    & = |E[S]| - |S| - x_v&\\
    & = (|E[S-u]| + 1) - (|S - u| + 1) - x_v&\\
    & = b(S - u) - x_v&\\
    &\leq g_x(S - u) - x_v&
\end{align*}
The above chain of inequalities implies that $x_v \leq 0$. By the non-negativity constraints on $x_v$, we have that $x_v = 0$. Thus, the final inequality in the above chain is in fact an equality, and hence, the set $S-u$ is in $\calT_x$. 

We now make three observations: 
\begin{enumerate}
    \item $\row(S) - \chi_v = \row(S - u)$ (we note that the LHS subtracts the indicator vector of the vertex $v$ while the RHS considers the row vector of the set $S-u$, where $u$ is the vertex with $d_S(u)=1$),
    \item Every $X \in\calC$ such that $X\subset S$ either has $d_X(u) = 1$ or $X$ does not contain $u$, and 
    \item For every $X\subset S$ such that $X\in\calC$ and $u,v \in X$, $\row(X) - \chi_v = \row(X - u)$. 
\end{enumerate} 
Here, the first observation is due to $d_S(u) = 1$ and thus $\row(S)_u = 0$. The second observation follows because of monotonicity of the induced-degree function and the fact that tight sets cannot have isolated vertices by \Cref{lem:tight-set-connected}. The third observation follows by writing down the corresponding chain of inequalities as above for each such set $X$. 
Then, by the three observations, we can remove the vertex $u$ from every set in $\calC$ that contains $u$ resulting in a set family $\calC' \not = \calC$ that still satisfies properties (1)-(4). However, we have that $\sum_{S\in\calC'}|S| < \sum_{S\in\calC}|S|$, contradicting our choice of family $\calC$.

Finally, we show that $\calC$ also satisfies property (6). 
By way of contradiction, assume that there exists $A, B\in \calC$ such that $A\subset B$ and $x_u = 0$ for each $u \in B - A$. We first show a lower bound on $|\delta(A, B-A)|$.
\begin{claim}\label{claim:non-zero-difference-of-tight-sets}
$\frac{|\delta(A, B-A)|}{2} \geq -b(B - A)$.
\end{claim}
\begin{proof}
First, we argue that the subgraph $G[B - A]$ is a forest. By way of contradiction, suppose that the subgraph $G[B-A]$ contains a cycle. Let $C$ be a minimal cycle in the subgraph $G[B-A]$. Then, $g_x(C) = 0$ as $x_u = 0$ for each $u \in B - A$ by our hypothesis. Furthermore, $b(C) = 0$ as $C$ is a cycle. Thus, $C \in \calT_x$ is a tight set. However, since $C\subseteq B - A$, we have that $C\cap A = \emptyset$. This contradicts the fact that all tight sets must overlap (as shown by \Cref{lem:tight-sets-conditional-uncrossing-properties}).

Let $k$ be the number of disconnected acyclic components of the forest $G[B - A]$. Then, we have the following two observations: 
\begin{enumerate}
    \item $b(B - A) = -k$ and 
    \item $|\delta(A, B - A)| \geq 2k$.
\end{enumerate}
We note that the claim  
follows from above observations. We justify these observations next. 

The first observation follows from the fact that the subgraph $G[B-A]$ is a forest. In particular, we have that $$b(B-A) = \left|E[B-A]\right| - \left|B - A\right| = \left(\left(\left|B-A\right| - 1\right) - (k-1)\right) - |B - A| = -k.$$
We now prove the second observation. Each of the $k$ acyclic components of the subgraph $G[B - A]$ is either a singleton or has at least two leaves. If the component is a singleton vertex $v$, then we have that $|\delta(v, A)| = d_B(v) \geq 2$, where the inequality is because $B \in \calC$ and the family $\calC$ satisfies property (5). Alternatively, suppose that the component has two leaves. Let $v$ be an arbitrary leaf of the component. We note that the induced degree $d_{B-A}(v) = 1$. Thus, $|\delta(v, A)| = d_B(v) - d_{B-A}(v) \geq 1$, where the inequality is once again because $B \in \calC$ and $\calC$ satisfies property (5). 
\end{proof}
With the above lower bound on the cut size $|\delta(A, B-A)|$, we get the required contradiction as follows:
\begin{align*}
    g_x(A) + \frac{|\delta(A, B-A)|}{2} \ & > g_x(A) + \sum_{uv\in\delta(A, B-A)}x_u&\\
    & = g_x(A) + g_x(B-A) + \sum_{uv\in\delta(A, B-A)}(x_u + x_v)&\\
    & = g_x(B)&\\
    &= b(B)&\\
    & = |E[B]| - |B|&\\
    &= \left( |E[A]| + |E[B - A]| + |\delta(A, B-A)|\right) - \left(|A| + |B - A|\right)&\\
    &=b(A) + b(B - A) + |\delta(A, B-A)|&\\
    &\geq b(A) + \frac{|\delta(A, B-A)|}{2}&\\
    & = g_x(A) + \frac{|\delta(A, B-A)|}{2}.&
\end{align*}
The first inequality follows from $x_u < 1/2$ $\forall u\in V$. The first equality follows from the assumption that $x_u = 0$ for each $u \in B - A$. The penultimate inequality follows from \Cref{claim:non-zero-difference-of-tight-sets}, and the final equality holds since the set $A\in \calT_x$ is a tight set.
\end{proof}

\subsubsection{Proof of Theorem \ref{thm:weak-density-extreme-point}}\label{sec:counting-argument}
We now complete the proof of \Cref{thm:weak-density-extreme-point}.  
In the next lemma, we use a stronger hypothesis that $x_u<1/3$ for all $u\in V$ to 
show that there exist at least two vertices $u, v$ which take non-zero values. We emphasize that the lemma does not rely on extreme point properties and holds for every feasible point $x$ satisfying the hypothesis.  

\begin{lemma}\label{lem:two-non-zero-vertices-general-statement}
Let $G=(V, E)$ be a connected graph with minimum degree at least $2$ such that $G$ contains a $2$-pseudotree. Let $x\in \weakdensitypolyhedron(G)$ such that $x_u<1/3$ for all $u\in V$. Then, there exist distinct vertices $u, v\in V$ such that $x_u, x_v>0$.
\end{lemma}
\begin{proof}
By way of contradiction, let $S\subseteq V$ be a set such that the subgraph $G[S]$ is connected, has minimum-degree at least $2$, contains a $2$-pseudotree, but 
$S$ contains at most one vertex with a corresponding non-zero vertex variable.
We first consider the case where $x_v = 0$ for all $v \in S$. 
Observe that in this case, $g_x(S) = 0$. Consequently, we have that $0 = g_x(S) \geq b(S) = |E[S]| - |S| \geq 1$, a contradiction. Here, the first inequality is by the weak-density constraint on $S$, and the second inequality is because $G[S]$ is connected and contains a $2$-pseudotree---this can be observed by starting with the $2$-pseudotree, and then charging the remaining vertices to the edges which connect them to the $2$-pseudotree.

We next consider the case where the set $S$ contains exactly one vertex $v\in S$ such that $x_v > 0$. It follows that 
$$(d_S(v)-1)x_v = g_x(S) \geq b(S) = |E[S]| - |S|.$$
Here the first equality is because $x_v$ is the only non-zero variable, while the first inequality is by the weak-density constraint on $S$. We now consider two cases based on the degree of $v$ in $G[S]$.

First, consider the case where $d_S(v) \geq 4$. Then we have the following:
$$x_v = \frac{|E[S]| - |S|}{d_S(v) - 1} \geq \frac{|S| - 1 + \frac{d_S(v)}{2} - |S|}{d_S(v) - 1}= \frac{1}{2}\left(1 - \frac{1}{d_S(v) - 1}\right)\geq \frac{1}{2}\cdot\frac{2}{3} = \frac{1}{3},$$
a contradiction to our hypothesis that $x_u < 1/3$ $\forall u \in V$.
Here, the first inequality is by hypothesis that $d_S(u) \geq 2$ for each $u \in S$.

Next, consider the case where $d_S(v) \leq 3$. We recall (from the last sentence of the first paragraph) that $b(S) = |E[S]| - |S| \geq 1$ since $G[S]$ is connected and contains a $2$-pseudotree. This gives us the required contradiction as follows:
$$1 \leq b(S) \leq g_x(S) = (d_S(v)-1)x_v \leq 2x_v < 1,$$
Here, the first inequality holds by the weak-density constraints on $S$, while the the final inequality holds due to our hypothesis that $x_v < \frac{1}{2}$.

\end{proof}

We use \Cref{lem:two-non-zero-vertices-general-statement} to conclude the following: 

\begin{corollary}\label{cor:two-non-zero-vertices-in-tight-sets}
Let $G=(V, E)$ be a connected graph with minimum degree at least $2$ such that $G$ contains a $2$-pseudotree. 
Let $x$ be an extreme point of $\weakdensitypolyhedron(G)$ such that $x_u < \frac{1}{3}$ for all $u\in V$, the family of tight sets for $x$ be $\calT_x := \{S\subseteq V : f_x(S) = 0\}$, and let $\calZ:= \left\{u \in V : x_u = 0\right\}$. 
Let $\calC\subseteq\calT_x$ be a chain family satisfying the properties in \Cref{lem:chain-basis-existence}. 
Then, 
for every $S\in \calC$, there exist distinct vertices $u,v \in S$ such that $x_u, x_v > 0$.
\end{corollary}
\begin{proof}
By \Cref{lem:chain-basis-existence}, we have that for every $A\in\calC$, the subgraph $G[A]$ contains a $2$-pseudotree. The claim follows by \Cref{lem:two-non-zero-vertices-general-statement}.
\end{proof}

We now restate and prove Theorem \ref{thm:weak-density-extreme-point}.

\thmWeakDensityExtremePoint*

\begin{proof}
By way of contradiction, let $x_u < \frac{1}{3}$ for all $u \in V$. 
Let $\calC =\{C_1, C_2, C_3, \ldots, C_t\}$ be the chain basis guaranteed by \Cref{lem:chain-basis-existence}. 
We order the sets in the basis so that
 $C_0 \subseteq C_1\subseteq C_2 \subseteq C_3 \subseteq \cdots\subseteq C_t \subseteq C_{t+1}$, where we denote $C_0 := \emptyset$ and $C_{t+1} := V$ for notational convenience. Let $\calN\calZ(C_i) = \{v \in C_i : x_v > 0\}$.
 Let $\calN\calZ := \calN\calZ(V)$ be the set of vertices $v$ with $x_v>0$. Then, $\calN\calZ$ can be partitioned as follows: $\calN\calZ(V) = \uplus_{i \in [t+1]} \calN\calC(C_i) \backslash \calN\calC(C_{i-1})$.
Hence, we have that 
\[
|\calC| < \sum_{i\in [t]} \left|\calN\calZ(C_i) \backslash \calN\calZ(C_{i-1})\right|\leq |\calN\calZ|.
\]
The first inequality holds due to the following two reasons: (i)$|\calN\calZ(C_1) \backslash \calN\calZ(C_0)| = |\calN\calZ(C_1)| \geq 2$ by \Cref{cor:two-non-zero-vertices-in-tight-sets} and (ii) $|\calN\calZ(C_i) \backslash \calN\calZ(C_{i-1})| \geq 1$ for $i \in [2,t]$ by property (6) of \Cref{lem:chain-basis-existence}. The second inequality is because $|\calN\calZ(C_{t+1}) \backslash \calN\calZ(C_{t})| \geq 0$ and our observation that the set of differences of subsequent basis sets in the chain ordering of $\calC$ partitions the set of non-zero variables.
Thus, we have that 
\[
|V| 
= |\calZ|+|\calC|
< |\calZ|+|\calN\calZ|
= |V|,
\]
a contradiction. The first equality is because $x$ is an extreme point and  $\calC$ is such that $\Rows(\calC\cup \calZ)$ is linearly independent and $\spanfunc(\Rows(\calC\cup \calZ))=\spanfunc(\Rows(\calT_x \cup \calZ))$ by \Cref{lem:chain-basis-existence}. The last equality is because the number of variables is equal to the sum of the number of zero variables and the number of non-zero variables. 

\end{proof}


\subsection{Extreme point property of the Orientation polyhedron}\label{sec:orientation-polytope-extreme-point}
In this section, we prove \Cref{thm:orientation-polyhedron-extreme-point}. We restate it below. 
\thmOrientationPolyhedronExtremePoint*

We recall that an extreme point of $\orientationpolyhedron(G)$ is said to be {\em minimal} if we cannot lower any single variable, keeping the others unchanged, while maintaining feasibility. 
Before proceeding with the proof of \Cref{thm:orientation-polyhedron-extreme-point}, we establish a few lemmas. In the lemmas below, we always assume that $G$ is not a pseudoforest, $(\bar{x},\bar{y})$ is a minimal extreme point of $\orientationpolyhedron(G)$ and, aiming toward a contradiction, $\bar{x}(v) < 1/3$ for all vertices $v \in V(G)$.

For a graph $G$, we denote its edge-vertex incidence graph by $H$. Thus, $V(H) = V(G) \cup E(G)$ and $E(H) = \{ev : e \in E(G),\ v \in V(G),\ e \in \delta(v)\}$. We note that $|V(H)| = |V(G)| + |E(G)|$ and $|E(H)| = 2 |E(G)|$. The \emph{support graph} of $\bar{y}$ is the subgraph of the incidence graph $H$ whose vertices are all the vertices of the incidence graph $H$ and whose edges are those for which $\bar{y}_{e,v} > 0$. Denoting this subgraph by $H(\bar{y})$, we have $V(H(\bar{y})) = V(H) = V(G) \cup E(G)$ and $E(H(\bar{y})) = \{ev \in E(H) : \bar{y}_{e,v} > 0\}$.

\begin{lemma}
\label{lem:acyclic}
The support graph $H(\bar{y})$ is a forest.
\end{lemma}

\begin{proof}
Towards a contradiction, suppose that there exists a cycle in $H(\bar{y})$. Since $H(\bar{y})$ is bipartite, the edge set of this cycle can be partitioned into two matchings $M_1$ and $M_2$. Given some $\varepsilon > 0$, we define two points $\bar{y}^{\pm} \in \mathbb{R}^{E(H)}$ by letting $\bar{y}^{\pm} := \bar{y} \pm \varepsilon \chi^{M_1} \mp \varepsilon \chi^{M_2}$. If $\varepsilon$ is sufficiently small, both $(\bar{x},\bar{y}^+)$ and $(\bar{x},\bar{y}^-)$ are feasible (that is, belong to the orientation polytope), which contradicts the fact that $(\bar{x},\bar{y}) = \frac{1}{2} (\bar{x},\bar{y}^+) + \frac{1}{2} (\bar{x},\bar{y}^-)$ is extreme.
\end{proof}

\begin{lemma}
\label{lem:num_components}
The number of (connected) components of $H(\bar{y})$ is $|V(G)|-|E(G)|+m_0$, where $m_0$ denotes the number of edges $ev \in E(H)$ such that $\bar{y}_{e,v} = 0$.
\end{lemma}

\begin{proof}
We have that $|V(H(\bar{y}))|=|V(G)|+|E(G)|$ and $|E(H(\bar{y}))|=2|E(G)|-m_0$ by definition.  
By Lemma~\ref{lem:acyclic}, the number of components of $H(\bar{y})$ is 
\begin{align*}
|V(H(\bar{y}))| - |E(H(\bar{y}))| 
  &= (|V(G)| + |E(G)|) - (2 |E(G)| - m_0)\\
  &= |V(G)| - |E(G)| + m_0\,.
\end{align*}
\end{proof}

\begin{lemma}
\label{lem:edges}
Every edge $e = vw \in E(G)$, we have $\bar{y}_{e,v} > 0$ or $\bar{y}_{e,w} > 0$ (or both) and $\bar{x}_v + \bar{x}_w + \bar{y}_{e,v} + \bar{y}_{e,w} = 1$.
\end{lemma}

\begin{proof}
If both $\bar{y}_{e,v} = 0$ and $\bar{y}_{e,w} = 0$ then, since $(\bar{x},\bar{y})$ is feasible, we have $\bar{x}_v + \bar{x}_w \geq 1$. Hence, $\bar{x}_v \geq 1/2$ or $\bar{x}_w \geq 1/2$, a contradiction.

For the second part, toward a contradiction, suppose that $\bar{x}_v + \bar{x}_w + \bar{y}_{e,v} + \bar{y}_{e,w} > 1$. We may assume (by symmetry) that $\bar{y}_{e,v} > 0$. Then, slightly decreasing $\bar{y}_{e,v}$ preserves the feasibility of $(\bar{x},\bar{y})$. This contradicts the minimality of $(\bar{x},\bar{y})$.
\end{proof}

Now consider a component $T$ of the support graph $H(\bar{y})$. By Lemma~\ref{lem:acyclic}, the component $T$ is a tree. We define the \emph{defect} of $T$ as the number vertices $v \in V(T) \cap V(G)$ such that $\bar{x}_v > 0$, plus the number of vertices $v \in V(T) \cap V(G)$ such that $\bar{x}_v + \sum_{e \in \delta(v)} \bar{y}_{e,v} < 1$. Below, we denote this quantity by $\mathrm{defect}(T)$.

We say that component $T$ is \emph{tight} if $\bar{x}_v + \sum_{e \in \delta(v)} \bar{y}_{e,v} = 1$ for all vertices $v \in V(T) \cap V(G)$. Tight components play an important role in our analysis. We seek a tight component with extra properties, dubbed `interesting' (Lemma~\ref{lem:exists_interesting_comp} below states that such tight components exist).

\begin{lemma}
\label{lem:tight}
Let $T$ be a tight component of $H(\bar{y})$. Suppose that some vertex $v \in V(G)$ is a leaf of $T$. Then $T$ has exactly two vertices and $\bar{x}_w = 0$ for the other vertex $w \in V(G)$ that is incident to the unique edge of $G$ in $T$.
\end{lemma}

\begin{proof}
Since $T$ is tight, we have $\bar{x}_v + \sum_{e \in \delta(v)} \bar{y}_{e,v} = 1$. If $\bar{y}_{e,v} = 0$ for all $e \in \delta(v)$, we get $\bar{x}_v = 1$, a contradiction. Hence, $T$ has at least two vertices and we have $\bar{y}_{e,v} = 0$ for all $e \in \delta(v)$ except for precisely one edge, say $f = vw$. From $\bar{x}_v + \bar{y}_{f,v} = 1$ and $\bar{x}_v + \bar{x}_w + \bar{y}_{f,v} + \bar{y}_{f,w}$ (see Lemma~\ref{lem:edges}), we get $\bar{x}_w + \bar{y}_{f,w} = 0$, which implies $\bar{x}_w = 0$ and $\bar{y}_{f,w} = 0$. The result follows.
\end{proof}

We will call a tight component $T$ with exactly two vertices as a \emph{dyad}. Let $T$ be a dyad, say, with $V(T) = \{v,f\}$ where $v \in V(G)$ and $f = vw \in E(G)$. We note that $x_v + \sum_{e \in \delta(v)} y_{e,v} = 1$ follows from the tight constraints $x_v + x_w + y_{f,v} + y_{f,w} = 1$, $x_w = 0$, $y_{e,v} = 0$ for all $e \in \delta(v) \setminus \{f\}$ and $y_{f,w} = 0$. 

\begin{lemma} \label{lem:avg_defect_nonsimple}
Let $T_1$, \ldots, $T_k$ denote the components of $H(\bar{y})$ and let $d$ denote the number of dyads. We have
\begin{equation}
\label{eq:counting}
\sum_{i=1}^{k} \mathrm{defect}(T_i) \leq k-d\,.
\end{equation}
\end{lemma}

\begin{proof}
The number of variables defining 
the orientation polytope is $|V(G)| + |E(H)| = |V(G)| + 2 |E(G)|$. By Lemma~\ref{lem:edges}, we can write the total number of constraints that are tight at $(\bar{x},\bar{y})$ as $|E(G)| + n_\mathrm{tight} + n_0 + m_0$, where $n_\mathrm{tight}$ denotes the number of vertices $v \in V(G)$ such that $\bar{x}_v + \sum_{e \in \delta(v)} \bar{y}_{e,v} = 1$, $n_0$ denotes the number of vertices $v \in V(G)$ such that $\bar{x}_v = 0$, and $m_0$ (as above) denotes the number of edges $ev \in E(H)$ such that $\bar{y}_{e,v} = 0$. For every dyad $T$, the tight constraint $\bar{x}_v + \sum_{e \in \delta(v)} \bar{y}_{e,v} = 1$ follows from the other tight constraints, where $v$ is the unique vertex of $G$ in $T$. Since $(\bar{x},\bar{y})$ is an extreme point, we have that the number of variables is at most the number of tight constraints. Hence, 
\begin{align*}
|V(G)| + 2 |E(G)| &\leq |E(G)| + n_\mathrm{tight} + n_0 + m_0 - d.
\end{align*}
Rewriting the above gives
\begin{align*}
|V(G)| - n_\mathrm{tight} + |V(G)| - n_0 &\leq |V(G)| - |E(G)| + m_0 - d\,.
\end{align*}
Inequality \eqref{eq:counting} follows from this and Lemma~\ref{lem:num_components}.
\end{proof}

We call a component $T$ of $H(\bar{y})$ \emph{interesting} if it is tight, has defect at most $1$ and is not a dyad. By Lemma~\ref{lem:tight}, interesting components $T$ have the following properties: (i) $T$ has at least three vertices, (ii) every leaf of $T$ is an edge of $G$, and (iii) $\bar{x}(v) = 0$ for all vertices $v \in V(T) \cap V(G)$ except at most one.

\begin{lemma} \label{lem:exists_interesting_comp}
At least one component of $H(\bar{y})$ is interesting.
\end{lemma}

\begin{proof}
First, we observe that if $\bar{x}_v = 0$ for all vertices $v \in V(G)$, then $G$ is a pseudoforest, which contradicts our hypothesis. Hence, there exists a vertex $r \in V(G)$ with $\bar{x}_r > 0$. 

Let $T_1$, \ldots, $T_k$ denote the components of $H(\bar{y})$. 
If some $T_i$ which is not a dyad has $\mathrm{defect}(T_i) = 0$, then $T_i$ is interesting. Hence, we may assume that $\mathrm{defect}(T_i) \geq 1$ for all components that are not dyads. 
By \Cref{lem:avg_defect_nonsimple}, 
this implies $\mathrm{defect}(T_i) = 0$ for all dyads and $\mathrm{defect}(T_i) = 1$ for all the other components. Then, the unique component containing $r$ is interesting.
\end{proof}

We now restate and prove \Cref{thm:orientation-polyhedron-extreme-point}.
\thmOrientationPolyhedronExtremePoint*
\begin{proof}
Aiming toward a contradiction, suppose that $\bar{x}_v < 1/3$ for all vertices $v \in V(G)$. Hence, each one of the above lemmas apply. 
Let $T$ denote an interesting component of $H(\bar{y})$, which exists by Lemma~\ref{lem:exists_interesting_comp}. Since it is interesting, $T$ has at least two vertices and each leaf of $T$ is an edge of $G$. 

Let $r \in V(T) \cap V(G)$ denote any vertex such that $\bar{x}_v = 0$ for all vertices $v \in V(T) \cap V(G)$ distinct from $r$. We call $r$ the \emph{root} of $T$. For $w \in V(G)$, let $d_T(w)$ denote the number of edges $e \in V(T) \cap E(G)$ that are incident to $w$. (This is a slight abuse of notation since $d_T(w)$ is also defined for vertices $w \in V(G)$ outside of $T$.) We let $N(T)$ denote the set of vertices $w \in V(G)$ such that $d_T(w) > 0$ and $w \notin V(T)$. Finally, we let $\mathcal{L}(T) \subseteq E(G)$ denote the leaf set of $T$. 

By Lemma~\ref{lem:edges} and since $T$ is an interesting component of the support graph $H(\bar{y})$, we have
\begin{align*}
&|V(T) \cap E(G)| - |V(T) \cap V(G)|\\
&= \sum_{e = vw \in V(T) \cap E(G)} \left(\bar{x}_v + \bar{x}_w + \bar{y}_{e,v} + \bar{y}_{e,w}\right) 
- \sum_{v \in V(T) \cap V(G)} \left(\bar{x}_v + \sum_{e \in \delta(v)} \bar{y}_{e,v}\right)\\
&= (d_T(r) - 1) \bar{x}_r + \sum_{w \in N(T)} d_T(w) \bar{x}_w\,.
\end{align*}
We notice that $|V(T) \cap E(G)| - |V(T) \cap V(G)| = |\mathcal{L}(T)| - 1$. Hence, from the equations above, we get
\begin{equation}
\label{eq:component}
(d_T(r) - 1) \bar{x}_r + \sum_{w \in N(T)} d_T(w) \bar{x}_w = |\mathcal{L}(T)|-1\,.
\end{equation}

Let $M := (d_T(r)-1) + \sum_{w \in N(T)} d_T(w) = (d_T(r)-1) + |\mathcal{L}(T)|$ denote the sum of the coefficients in the left-hand side of \eqref{eq:component}. We note that $d_T(r) \leq |\mathcal{L}(T)|$. This implies $M \leq 2|\mathcal{L}(T)| - 1$. Thus,
$$
\frac{1}{M} \left( (d_T(r) - 1) \bar{x}_r + \sum_{w \in N(T)} d_T(w) \bar{x}_w\right) \geq \frac{|\mathcal{L}(T)|-1}{2|\mathcal{L}(T)|-1} \geq \frac{1}{3}\,.
$$
This is a contradiction since the left-hand side is an average of values that are all strictly less than $1/3$.
\end{proof}

\section{Feedback Vertex Set}
\label{sec:poly-sized-lp-relaxation-with-gap-atmost-two-for-fvs}
In this section, we prove our main results for FVS. We prove \Cref{thm:integrality-gap-wd+cycle-cover} in Section \ref{sec:integrality-gap-of-weak-density-and-cycle-cover} and \Cref{thm:poly-sized-LP-with-integrality-gap-atmost-2-for-FVS} in Section \ref{sec:orientation-and-cycle-cover-formulation}. We discuss two additional polynomial-sized ILP formulations for FVS with integrality gap at most $2$ in Section \ref{sec:more-FVS-formulations}. 







\subsection{Weak Density and Cycle Cover constraints}\label{sec:integrality-gap-of-weak-density-and-cycle-cover}
In this section, we restate and prove Theorem \ref{thm:integrality-gap-wd+cycle-cover}.

\thmWeakDensityCycleCover*

Let $G=(V, E)$ be a graph with non-negative vertex costs $c: V\rightarrow \R_{\ge 0}$. By definition, the ILP $\min\{\sum_{u\in V}c_ux_u: x\in \cyclecoverpolyhedron(G)\cap \Z^V\}$ formulates FVS. We note that weak density constraints are valid for FVS. Consequently, (\ref{FVS-IP: WD-and-cycle-cover}) is an ILP formulation for FVS. 
The rest of the section is devoted to proving the integrality gap bound mentioned in 
\Cref{thm:integrality-gap-wd+cycle-cover}. 
\begin{lemma}\label{lemma:weak-density-cycle-cover-integrality-gap}
    The integrality gap of the LP (\ref{FVS-LP:weak-density-cycle-cover}) is at most $2$. 
\end{lemma}

Let $\calC$ be the set of cycles in the input graph $G$.  Let $b(S):=|E[S]|-|S|$ for all $S\subseteq V$. 
We rewrite (\ref{FVS-LP:weak-density-cycle-cover}) and its dual in \Cref{fig:primal-dual-LPs}. 

\begin{figure}[H]
    \centering
    \begin{mdframed}
    \footnotesize
    \begin{tabular}{cc|cc}
    $\begin{array}{ll@{}ll}
        & \text{\textbf{Primal}:}  & &\\
        &&&\\ 
        \text{min}  &\displaystyle\sum\limits_{u \in V} c_{u}x_{u}&   &\\
        &&&\\ 
        \text{s.t.}& \displaystyle\sum\limits_{u\in S}(d_S(u)-1)x_u \geq b(S), & & \forall S \subseteq V\\
        &\displaystyle\sum\limits_{u\in C}x_{u} \geq 1 &    & \forall C \in \calC\\
        &x_u \geq 0&&\forall u \in V
    \end{array}$ 
    & & & 
    $\begin{array}{ll@{}ll}
        & \text{\textbf{Dual}:} & & \\
        &&&\\ 
        \text{max}  &\displaystyle\sum\limits_{S\subseteq V} b(S)y_S + \displaystyle\sum\limits_{C\in\calC}z_C &   &\\
        &&&\\ 
        \text{s.t.}& \displaystyle\sum\limits_{S\subseteq V:u\in S}(d_S(u)-1)y_S + \displaystyle\sum\limits_{C\in \calC:u\in C}z_C \leq c_u & & \forall u \in V\\
        &y_S \geq 0 &    & \forall S \subseteq V\\
        &z_C \geq 0&&\forall C \in \calC
    \end{array}$ 
    \end{tabular}
\end{mdframed}
    \caption{\textbf{WD-CC} Primal and Dual LPs for FVS}
  \label{fig:primal-dual-LPs}
\end{figure}

We will need the definition of \emph{semi-disjoint cycles} from \cite{CHUDAK1998111}:  
A cycle $C$ in graph G is \emph{semi-disjoint} (with respect to the graph $G$) if there is at most one vertex in $C$ of degree strictly larger than $2$. We call such a vertex as the \emph{pivot} vertex of the semi-disjoint cycle. We note that such cycles can be viewed as ``hanging off'' the graph via their pivots which are cut vertices. 

\subsubsection{Minimal FVS and Weak Density constraints}
In this section, we show an important property of inclusion-wise \emph{minimal} feedback vertex sets. Let $F$ be a feedback vertex set for a graph $G$. A \emph{witness cycle} for a vertex $f \in F$ is a cycle $C$ of the graph $G$ such that $F\cap C = \{f\}$ i.e.\ $f$ is the only vertex that \emph{witnesses} the cycle $C$. 
We note that every vertex of a minimal FVS must have a witness cycle.
In \Cref{lem:minimal-fvs-is-2apx}, we show that for a graph with non-trivial degree and no semi-disjoint cycles, every \emph{minimal} feedback vertex set is essentially a $2$-approximation to the optimal feedback vertex set with respect to an appropriate cost function. We note that our proof of the lemma appears implicitly in \cite{CHUDAK1998111}. However, we include it here for completeness.

\begin{lemma}\label{lem:minimal-fvs-is-2apx}
Let $G$ be a graph with minimum-degree at least $2$ and containing no semi-disjoint cycles. Then, for every minimal feedback vertex set $F$, we have that $$\sum_{v\in F}(d(v) - 1) \leq 2b(V).$$
\end{lemma}
\begin{proof}
We first show a convenient sufficient condition that proves the claimed upper bound. We subsequently prove this sufficient condition by an edge-counting argument. Let $t$ denote the number of connected components in $G - F$. 

\begin{claim}\label{claim:lower-bound-cut-fvs}
If $|\delta(F, V-F)| \geq |F| + 2t$, then $\sum_{v \in F} (d(v) - 1)\leq 2b(V)$.
\end{claim}
\begin{proof}
We have the following:
\begin{align*}
    \sum_{v \in F}(d(v) - 1)& = 2|E| - |F| - \sum_{v \in V - F}d(v) &\\
    &=  2|E| - |F| - \left(|\delta(F, V - F)| + 2|E[V - F]|\right) &\\
    & =  2|E| - |F| - \left(|\delta(F, V - F)| + 2(|V| - |F| - t)\right)&\\
    &\leq  2|E| - |F| - \left(2|V| - |F|\right)&\\
    & = 2b(V).&
\end{align*}
Here, the first equality holds because the sum of all vertex degrees is $2|E|$, the third equality is because $G - F$ is acyclic, and the inequality is due to the hypothesis that $|\delta(F, V-F)| \geq |F| + 2t$.
\end{proof}

We now show that a minimal feedback vertex set $F$ has at least $|F|+2t$ edges crossing it.  
For this, we consider the auxiliary bipartite graph $H = (K\cup F, E_H)$ as follows: Each vertex $k \in K$ corresponds to a connected component $C_k$ in $G - F$. For vertices $k \in K, f \in F$, the graph $H$ contains the edge $(k,f)$ if 
there exists a vertex $v \in C_k$ such that the edge $(v, f) \in E$ exists in the original graph $G$. 
Furthermore, we define the weight of an edge $(k, f) \in E_H$ as $w_H(k,f):=|\delta_G(C_k, f)|$. We note that $|K|=t$.

Since $F$ is an inclusionwise minimal feedback vertex set, every vertex $f \in F$ has a witness cycle. We note that there are no edges between the components of $G - F$. Thus, every witness cycle is completely contained in some component $C_k$ for some $k \in K$. In particular, this implies that in the graph $H$, every vertex $f\in F$ has an edge of weight at least $2$ incident on it. For each $f\in F$ pick an arbitrary such edge $e_f \in \delta_H(f)$ of weight at least $2$ and call it $f$'s \emph{primary} edge. Then, reducing the weight of all primary edges by $1$ still leaves all edges in $E_H$ with positive weight, but counts $|F|$ edges in the cut $|\delta(F, V-F)|$ in the original graph $G$. Let $H'$ denote the residual graph after this weight reduction. By \Cref{claim:lower-bound-cut-fvs}, it now suffices to show that the weight of all edges in $H'$ is at least $2|K|$.

For this, we claim that each vertex $k\in K$ must have a cumulative edge-weight of at least $2$ incident on it in the residual graph $H'$. By way of contradiction, suppose that there exists a vertex $k\in K$ with total incident edge-weight at most $1$. We consider two cases. First, suppose that $k$ has no incident edges. We recall that we reduced the weight of only primary edges in our weight-reduction step. Thus, if there was a primary edge incident to $k$, then $k$ would have a weight $1$ edge incident on it in the residual graph, contradicting the case assumption. Thus, there are no edges incident to $C_k$ in $G$, i.e. $C_k$ is a disconnected acyclic component of $G$. In particular, $C_k$ contains a leaf vertex, contradicting the hypothesis that the minimum degree of $G$ is at least $2$.

Next, consider the case where $k$ has a total weight of $1$ incident on it in the residual graph $H'$. Let $f \in F$ be its unique neighbor in the residual graph. Then, $f$ must also be its unique neighbor in $H$. We note that if $C_k$ has only one edge to $f$ in $G$, then $C_k$ must contain a leaf node in $G$ as it is acyclic, a contradiction. Thus, $C_k$ does not contain any leaf nodes, is acyclic, and has exactly $2$ edges to $f$ in $G$. In particular, $C_k \cup \{f\}$ is a semi-disjoint cycle in $G$, contradicting the hypothesis that $G$ has no semi-disjoint cycles. 
\end{proof}

\subsubsection{Proof of integrality gap}
In this section, we prove \Cref{lemma:weak-density-cycle-cover-integrality-gap}, 
i.e., we show that the integrality gap of the LP-relaxation for FVS given in \Cref{fig:primal-dual-LPs} is at most $2$. 
We prove this by constructing a dual feasible solution $(\overline{y}, \overline{z})$ and a FVS $F$ such that the cost of $F$ is at most twice the dual objective value of $(\overline{y}, \overline{z})$. This implies that the integrality gap of the LP is at most $2$ since the LP is a relaxation for FVS. We construct such a pair of solutions via a primal-dual algorithm. 
Our primal-dual algorithm is an extension of the $2$-approximation primal-dual algorithm due to \cite{CHUDAK1998111}. We state our algorithm in \Cref{alg:primal-dual-fvs}. 

\begin{algorithm}[H]
\caption{Primal-Dual for FVS}\label{alg:primal-dual}
\label{alg:primal-dual-fvs}
\textbf{Input:} (1) Graph $G=(V, E)$; (2) Vertex weights $c:V\rightarrow\mathbb{R_{+}}$.\\
\textbf{Output:} Feedback vertex set $F \subseteq V$.
\begin{enumerate}
    \item Initialize $(y,z) := 0$, $i := 1$, $F := \emptyset$, and $G_1 = G$.
    \item \textbf{While} the graph $G_i = (V_i, E_i)$ has a cycle:
    \begin{enumerate}
        \item Recursively remove vertices of degree $1$ from $G_i$.
        \item \textbf{If} $G_i$ has a semi-disjoint cycle $C_i$:\\ Define the set $S_i = C_i$ as the semi-disjoint cycle.\\
        Increase dual variable $z_{S_i}$ until the first dual constraint for some vertex $v_i \in C_i$ becomes tight.
        \item \textbf{Else}:\\ 
        Define the set $S_i = V_i$ as the entire vertex set.\\
        Increase dual variable $y_{S_i}$ until the first dual constraint becomes tight for some vertex $v_i \in V_i$.
        \item Update $F := F \cup \{v_i\}$\\
        $G_{i+1} := G_i - v_i$\\
        $i := i+1$.
    \end{enumerate}
    \item Perform reverse-delete on $F$.
    \item Return $F$.
\end{enumerate}
\end{algorithm}

\Cref{alg:primal-dual-fvs} is a primal-dual algorithm based on the LP relaxation for FVS given in \Cref{fig:primal-dual-LPs}. 
The algorithm initializes with the dual feasible solution $(\overline{y}=0, \overline{z}=0)$. In the $i$th iteration of the while-loop, the algorithm selects a specific set $S_i$ of vertices (which is either a cycle $C_i$ or entire vertex set $V_i$) and increases the corresponding dual variable until a dual constraint for some vertex $v_i \in S_i$ becomes tight. It then includes the vertex $v_i$ into the FVS and removes $v_i$ from the graph. Adding $v_i$ to the candidate set $F$ is interpreted as setting the primal variable $x_{v_i} = 1$. 
Thus, \Cref{alg:primal-dual-fvs} always maintains dual feasibility and primal complementary slackness (i.e., $x_{v_i}=1$ only if the dual constraint for $v_i$ is tight). 

We now explain the reverse-delete procedure mentioned in \Cref{alg:primal-dual-fvs}. 
Let $\ell$ be the number of iteration of the while loop. 
The reverse-delete procedure iteratively considers vertices in the reverse-order in which they were added into $F$, i.e. $v_{\ell}, v_{\ell-1}, \ldots v_1$. For every vertex in this order, the procedure checks whether removal of this vertex from $F$ is still a feasible FVS for $G$. If so, then the vertex is removed from $F$. We will later rely on reverse-delete to show that $F\cap S_i$ is a minimal FVS in $G[S_i]$ for every iteration $i$. 

We now observe that the algorithm terminates in polynomial time to return a feasible FVS. 

\begin{lemma}[Feasibility and Runtime]\label{lem:primal-dual-fvs-feasability}
\Cref{alg:primal-dual-fvs} returns a feasible FVS for the input graph $G$ in polynomial time.
\end{lemma}
\begin{proof}
Each step of the while loop and reverse-delete procedure can be implemented to run in polynomial time. 
In $i$th iteration of the while-loop, at least one vertex of the graph $G_i$ is removed. The empty graph has no cycles and so the while-loop will terminate in at most linear number of iterations. Furthermore, the reverse delete step considers every vertex in $F$ exactly once, and so, it also terminates in at most linear number of steps. The output set $F$ is a feasible FVS by the terminating condition of the while-loop and the fact that reverse-delete deletes a vertex only if its deletion maintains feasibility. 
\end{proof}

We now bound the approximation factor of the solution $F$ returned by \Cref{alg:primal-dual-fvs}. Let $(\overline{y}, \overline{z})$ be the dual solution constructed by \Cref{alg:primal-dual-fvs} and let $\chi^F\in \{0,1\}^V$ be the indicator vector of $F$. We note that \Cref{lem:primal-dual-fvs-feasability} implies that $\chi^F$ is a feasible solution to the primal LP. Let $\mathtt{primal}(\chi^F), \mathtt{dual}(\overline{y}, \overline{z})$ denote the objective values of the primal and dual LPs for solutions $\chi^F$ and $(\overline{y}, \overline{z})$ respectively.

Let $I_1\subseteq [\ell]$ denote the set of iterations in which the dual variable for the entire residual graph was increased (i.e. statement (c) of \Cref{alg:primal-dual-fvs} is executed) and let  $I_2\subseteq [\ell]$ denote the set of iterations in which the dual variable for a semi-disjoint cycle in the residual graph was increased (i.e. statement (b) of \Cref{alg:primal-dual-fvs} is executed).
For notational convenience, we let $F_j = F \cap \{v_j\}$ and $F_{\geq j} = \cup_{k \in[j,\ell]}F_j$.
The next lemma shows that the set $S_i\cap F_{\geq i}$ is  a minimal feedback vertex set for the subgraph $G[S_i]$ for each $i \in [\ell]$. 

\begin{lemma}\label{lem:F-minimal-for-each-set-of-algorithm}
The set $S_i \cap F_{\geq i}$ is a minimal feedback vertex set for the subgraph $G[S_i]$ for each $i\in [\ell]$.
\end{lemma}
\begin{proof}
We prove by induction on $\ell-i$. 
For the base case, we consider $i = \ell$. We observe that $F_{\geq \ell} = \{v_{\ell}\}$ as otherwise, the following two observations result in a contradiction: (i) The reverse-delete step only removes $v_\ell$ from $F$ if the set $F - \{v_{\ell}\}$ is a feasible feedback vertex set for $G$; and (ii) if $F - \{v_{\ell}\}$ were a FVS for $G$, then the algorithm would have terminated after the $(\ell - 1)^{th}$ iteration. We note that $F_\ell$ is a FVS for $G_\ell$ as there was no $(\ell+1)^{th}$ iteration. If the set $S_\ell = V_\ell$ is the entire vertex set of the graph $G_\ell$, then $v_\ell \in V_\ell$.
Alternatively, if the set $S_\ell$ is a semi-disjoint cycle $C_\ell$, then $v_\ell \in C_\ell$, as otherwise the cycle $C_\ell$ would remain in $G_{\ell+1}$ contradicting that there were only $\ell$ iterations. In both scenarios, the set $S_{\ell}\cap F_{\geq\ell}$ is a minimal FVS for $G[S_\ell]$.

For the inductive case, assume that $0 \leq i < \ell$. By the inductive hypothesis, we have that the set $S_j\cap F_{\ge j}$ is a minimal FVS for the graph $G[S_j]$ for all $\ell \geq j > i$. We consider two cases based on whether the set $S_i$ is a semi-disjoint cycle of $G_i$ or the entire vertex set of $G_i$.

First, we consider the case where $S_i$ is a semi-disjoint cycle $C_i$. We need to show that $|F\cap C_i|=1$. 
We observe that $1\leq |F \cap C_i| \leq 2$: The lower bound follows from the fact that the set $F_{\geq i}$ is a feasible FVS for $G_i$ and thus intersects all cycles in at least one vertex. The upper bound holds by the following three observations: (i) The vertex $v_i$ is the first vertex of the set $C_i$ to be selected into $F$ by the algorithm; (ii) If the vertex $v_i$ is the pivot vertex, then all vertices in $C_i - \{v_i\}$ get recursively removed and thus cannot be selected into $F$ in any subsequent iteration after $i$; and (iii) If the vertex $v_i$ is not the pivot vertex, then all vertices in $C_i - \{v_i\}$ except the pivot vertex get recursively removed due to step (a) and thus, only the pivot vertex can be included into $F$ in some subsequent iteration after $i$. We note that if scenario (ii) happens, then $F_{\geq i}\cap C_i = \{v_i\}$ and hence, $|F\cap C_i|=1$. Alternatively, suppose that scenario (iii) happens. If the pivot is never selected into $F$, then once again we have $F_{\geq i}\cap C_i = \{v_i\}$ and hence, $|F\cap C_i|=1$. Otherwise the pivot is selected into $F$ and the reverse-delete step processes the pivot vertex before processing $v_i$. If the reverse-delete step removes the pivot vertex from $F$, then once again we have $F_{\geq i}\cap C_i = \{v_i\}$ and hence, $|F\cap C_i|=1$. Otherwise, the reverse-delete step does not remove the pivot from $F$. Consider the iteration of reverse-delete that processes $v_i$---we need to show that reverse-delete will indeed remove $v_i$ in this iteration. We observe that the set $\{v_k: k\in [1, i-1]\}$ has not been processed by reverse-delete and this set intersects every cycle that is not entirely contained in $G_i$ (this is true since $G_i$ is the residual graph after removal of $\{v_1, \ldots v_{i-1}\}$ along with recursively removing degree-1 vertices in each step). Since the pivot for $C_i$ is already in $F$ and the cycle $C_i$ is semi-disjoint in $G_i$, the reverse-delete step must remove the vertex $v_i$ from $F$ as $v_i$ intersects no cycles other than those already intersected by $F - \{v_i\}$. In particular, $F_{\geq i}\cap C_i $ is a singleton set consisting of the pivot vertex for $C_i$. Thus, in all cases we have that $|F_{\geq i}\cap C_i| = 1$. It follows that $F_i$ is a minimal FVS for the cycle $G[C_i]$.

Next, we consider the case when $S_i = V_i$ is the entire vertex set of the subgraph $G_i$. 
Here, we consider the set of all cycles that $v_i$ intersects in $G$ and partition them into two types: (1) Cycles completely contained in $G_i$; and (2) cycles that include at least one vertex from $V-V_i$. 
Since $F_{\geq i}$ is a feasible FVS for $G_i$, we have that $F_{\geq i}$ intersects all cycles of the first type.
Furthermore, since the vertices in $V- V_i$ do not exist in $G_i$, the set $\{v_1,\ldots , v_{i-1}\}$ intersect all cycles of the second type. We recall that the set $F_{\geq i+1}$ is a minimal FVS for $G_{i+1}$. 

We consider the first subcase where the  set $F_{\geq i+1}$ is an FVS for the graph $G_i$. Then, the set $F_{\geq i+1}$ intersects all cycles contained in $G_{i}$. In particular, it intersects all cycles of type (1). Thus, the set $\{v_1,\ldots, v_{i-1}\} \cup F_{\geq i+1}$ is a feasible FVS for the original graph $G$. Due to this, the reverse-delete step removes the vertex $v_i$ from $F$ and we have that $F_{\geq i} = F_{\geq i+1}$. Furthermore, by the inductive hypothesis we have that $F_{\geq i+1}$ is a minimal FVS for $G_{i+1}$. Thus, the set $F_{\geq i}$ is a minimal FVS for $G_i = G[S_i]$.

We next consider the remaining subcase where the  set $F_{\geq i+1}$ is not a FVS for the graph $G_i$. In particular, there must be a cycle of the first type that $v_i$ intersects, but none of the vertices of $F_{\geq i+1}$ intersect. In particular, the set $\{v_1,\ldots, v_{i-1}\} \cup F_{\geq i+1}$ is not a feasible FVS for the original graph $G$ as the set does not intersect all type (1) cycles. Thus, the reverse-delete step cannot remove $v_i$ from $F$. By the induction hypothesis, we have that $F_{\geq i+1}$ is a minimal FVS for $G_{\geq i+1}$. Consequently, none of these vertices can be removed from $F_{\geq i}$ as the resulting set would not be a feasible FVS. It follows that $F_{\geq i}$ is a minimal FVS for $G_i = G[S_i]$.
\end{proof}

We show in \Cref{lemma:WD-CC-LP-integrality-gap} that the solution $F$ constructed by the algorithm has cost at most twice the objective value of the dual feasible solution $(\overline{y}, \overline{z})$ constructed by the algorithm. We recall that $\chi^F\in \{0,1\}^V$ is the indicator vector of the set $F$ and $\mathtt{primal}(\chi^F), \mathtt{dual}(\overline{y}, \overline{z})$ denote the objective values of the primal and dual LPs for solutions $\chi^F$ and $(\overline{y}, \overline{z})$ respectively. 

\begin{lemma}[Integrality gap bound]\label{lemma:WD-CC-LP-integrality-gap}
We have that 
\[
\mathtt{primal}(\chi^F) 
\le  2\cdot\mathtt{dual} (\overline{y}, \overline{z}).
\]
\end{lemma}
\begin{proof} We have that 
\begin{align*}
    \mathtt{primal}(\chi^F)& = \sum_{v\in F} c_v&\\
    & = \sum_{v\in F}\left(\sum_{S\subseteq V:v\in S}(d_S(v) - 1)\overline{y}_S + \sum_{C\in \calC:v \in C}\overline{z}_C\right)&\\
    &= \sum_{S\subseteq V}\overline{y}_S\sum_{v \in F\cap S}(d_S(v) - 1) + \sum_{C\in\calC}\overline{z}_C\sum_{v \in F\cap C}1&\\
    &= \sum_{i\in I_1}\overline{y}_{S_i}\sum_{v \in F\cap S_i} (d_{S_i}(v) - 1) + \sum_{i \in I_2}\overline{z}_{C_i}|F \cap C_i|&\\
    &\leq 2\sum_{i \in I_1}\overline{y}_{S_i}b(S_i) + \sum_{i \in I_2}\overline{z}_{C_i} &\\
    &\leq 2\cdot\mathtt{dual} (\overline{y}, \overline{z}).&
\end{align*}
Here, the second equality follows by the primal complementary slackness maintained by \Cref{alg:primal-dual-fvs}. The fourth equality follows by the fact that all dual variables that were not incremented during some iteration of \Cref{alg:primal-dual-fvs} have value zero. The first inequality follows by \Cref{lem:F-minimal-for-each-set-of-algorithm} and \Cref{lem:minimal-fvs-is-2apx}. Here, we note that for $i \in I_2$, $|F\cap C_i| = 1$ (by \Cref{lem:F-minimal-for-each-set-of-algorithm}) because $G[C_i]$ is an induced cycle and consequently a minimal FVS for $G[C_i]$ contains exactly $1$ vertex.
\end{proof}

\Cref{lemma:WD-CC-LP-integrality-gap} completes the proof of \Cref{lemma:weak-density-cycle-cover-integrality-gap}
since $\mathtt{dual} (\overline{y}, \overline{z})\le \mathtt{primal}(\chi^F)$ by weak duality and the fact that the ILP $\min\{\sum_{u\in V}c_u x_u: x\in \weakdensitypolyhedron(G)\cap \cyclecoverpolyhedron(G)\cap \Z^V\}$ is an ILP formulation for FVS (whose LP-relaxation is given in \Cref{fig:primal-dual-LPs}). It also proves that the approximation factor of the solution $F$ returned by \Cref{alg:primal-dual-fvs} is at most $2$.

\subsection{Orientation and Cycle Cover constraints}\label{sec:orientation-and-cycle-cover-formulation}
In this section, we give a polynomial-sized ILP for FVS based on orientation and cycle cover constraints. Moreover, we show that the integrality gap of the LP-relaxation of this formulation is at most $2$. 
We consider the following formulation: 

\begin{align}
\min\left\{\sum_{u\in V}c_u x_u: x\in \projectedorientationpolyhedron(G)\cap \cyclecoverpolyhedron(G)\cap \Z^V \right\}. \tag{FVS-IP: orient-and-cycle-cover} \label{FVS-IP: orient-and-cycle-cover}
\end{align}

\begin{lemma}\label{lemma:integrality-gap-orient-intersect-cycle-cover}
For an input graph $G=(V, E)$ with non-negative vertex costs $c: V\rightarrow \R_{\ge 0}$, \eqref{FVS-IP: orient-and-cycle-cover} is an integer linear programming formulation for FVS. Moreover, its LP-relaxation has integrality gap at most $2$. 
\end{lemma}
\begin{proof}
The indicator vector $x$ of a feedback vertex set of $G$ is contained in $\projectedorientationpolyhedron(G)\cap\cyclecoverpolyhedron(G)$. Moreover, every integral solution $x\in \projectedorientationpolyhedron(G)\cap\cyclecoverpolyhedron(G)$ is the indicator vector of a feedback vertex set: since $x\in \projectedorientationpolyhedron(G)\cap \Z^V$, we have that $x\in \{0,1\}^V$ and since $x\in \cyclecoverpolyhedron(G)$, we have that $x$ is the indicator vector of a feedback vertex. 

We now show that the integrality gap of the LP-relaxation is at most $2$. 
By \Cref{lemma:orientation-formulation}(1), we have that $\projectedorientationpolyhedron(G)\subseteq \weakdensitypolyhedron(G)$ and hence, the integrality gap of the LP-relaxation is at most the integrality gap of (\ref{FVS-LP:weak-density-cycle-cover}). By Theorem \ref{thm:integrality-gap-wd+cycle-cover}, the integrality gap of (\ref{FVS-LP:weak-density-cycle-cover}) is at most $2$. 
\end{proof}

\Cref{thm:poly-sized-LP-with-integrality-gap-atmost-2-for-FVS}  follows from \Cref{lemma:integrality-gap-orient-intersect-cycle-cover}, the observation that  $\projectedorientationpolyhedron(G)$ can be expressed using polynomial number of variables and constraints (see the description of $\orientationpolyhedron(G)$), and the fact that the cycle cover polyhedron can equivalently be described using polynomial number of variables and constraints. We note that this fact is folklore, but we include a proof of it in \Cref{lemma:cycle-cover-poly-sized-lp} in the appendix.

\subsection{Alternative formulations}\label{sec:more-FVS-formulations}
%
In Section \ref{sec:orientation-and-cycle-cover-formulation}, we described an ILP formulation based on orientation and cycle cover polyhedra and showed that its LP-relaxation is at least as strong as (\ref{FVS-LP:weak-density-cycle-cover}). We briefly discuss two more polynomial-sized LP-relaxations for FVS with integrality gap at most $2$. We give full details in the appendix for the sake of completeness. 

We recall that Chekuri and Madan \cite{chekuri-madan16} gave a polynomial-sized ILP formulation for Subset-FVS. We note that the LP-relaxation of their ILP-formulation specialized for FVS has integrality gap at most that of the LP-relaxation of  \eqref{FVS-IP: orient-and-cycle-cover}. 
Consequently, the LP-relaxation of their ILP formulation has integrality gap at most $2$. We present the details in Appendix \ref{sec:CM-formulation}. This result gives additional impetus to improving the integrality gap analysis of their LP-relaxation for Subset-FVS. 

Next, we give an ILP formulation based only on orientation constraints (without relying on cycle cover constraints) - see Appendix \ref{sec:orientation-formulation}. We show that its LP-relaxation is at least as strong as that of the strong density constraints based formulation and consequently is also at most $2$. 


\paragraph{Acknowledgements.} 
Samuel Fiorini and Stefan Weltge would like to thank Gwena\"{e}l Joret and Yelena Yuditsky for preliminary discussions. Karthekeyan Chandrasekaran and Samuel Fiorini would like to acknowledge the support of the trimester program on Discrete Optimization at the Hausdorff Institute for Mathematics (2021) for facilitating initial discussions on the problem. Chandra Chekuri, Samuel Fiorini and Stefan Weltge would like to thank the Oberwolfach Research Institute for Mathematics and the organizers of the Combinatorial Optimization workshop (2021) during which this research started. Shubhang Kulkarni thanks Da Wei Zheng for engaging in preliminary discussions on the separation oracle for $2$-pseudotree cover constraints. We thank the reviewers for their comments which helped improve the presentation of the paper.

\section*{Declarations}
\paragraph{Funding and/or Conflicts of Interests/Competing Interests.}
Karthekeyan Chandrasekaran and Shubhang Kulkarni were supported in part by NSF grants CCF-1814613 and CCF-1907937. Chandra Chekuri was supported in part by NSF grants CCF-1910149 and CCF-1907937. Stefan Weltge was supported in part by Deutsche Forschungsgemeinschaft (DFG, German Research Foundation), project number 451026932. No Conflicts of Interests/Competing Interests. 

\bibliography{references.bib} 
\bibliographystyle{alpha}
\appendix
\section{Pseudoforest Deletion Set}

\subsection{Polynomial-time separation oracle for $2$-Pseudotree Cover constraints }\label{sec:2pt-cover-constraints-separation-oracle}
In this section, we show that the family of $2$-pseudotree cover constraints admits a polynomial time separation oracle.

\begin{restatable*}{lemma}{NWtwoPTPolytime}
\label{thm:MC2PT-polytime:main}
The family of $2$-pseudotree cover constraints admits a polynomial-time separation oracle. 
\end{restatable*}

In order to solve the separation problem, it suffices to design a polynomial time algorithm for the \emph{Minimum Cost $2$-Pseudotree} problem:  The input to MC2PT is a vertex-weighted graph $\left(G = \left(V,E\right), w:V\rightarrow\R_{\geq 0}\right)$, and the goal is to compute a minimum weight subset $U\subseteq V$ of vertices such that $G[U]$ is a $2$-pseudotree, i.e., $$\min\left\{\sum \nolimits_{u\in U}w(u): U\subseteq V \text{ and } G[U] \text{ is a }2\text{-pseudotree}\right\}.$$ 
Our strategy to solve MC2PT is to reduce it to solving a polynomial number of special instances of the \emph{Minimum Node-Weighted Steiner Tree}\footnote{We refer to \emph{nodes} as \emph{vertices} for consistency with the rest of our technical sections.} (NWST), and use the result of \cite{NWST-Buchanan-et.al} which says that these special instances of NWST can be solved in polynomial time. Formally, the input to NWST is a vertex-weighted graph $\left(G = \left(V,E\right), w:V\rightarrow\R_{\geq 0}\right)$ and a terminal set $S \subseteq V$. We will say that a graph $H$ is a \emph{Steiner tree} in $G$ for terminal set $S$ if $H$ is connected, acyclic, and is a subgraph of $G$ with $S\subseteq V_H$. Moreover, the weight of a subgraph $H$ of $G$ is the sum of weights of vertices in $H$. The NWST problem is to find a  minimum weight Steiner tree in $G$ for terminal set $S$, i.e., $\min\{\sum_{u\in V_H} w(u): H=(V_H, E_H) \text{ is a Steiner tree in } G \text{ for terminal set } S\}$.
NWST can be solved in polynomial time if the number of terminals is a constant.

\begin{proposition}[Theorem 1 of \cite{NWST-Buchanan-et.al}]\label{prop:NWST-polytime} There exists a $O(3^kn + 2^kn^2+n^3)$ time algorithm for NWST, where $k$ denotes the number of terminals and $n$ denotes the number of vertices in the input instance.
\end{proposition}

The following proposition shows a correspondence between minimum weight $2$-pseudotrees in a vertex-weighted graph and Steiner trees.

\begin{proposition}\label{prop:MC2PT-to-NWST}
    Let $G=(V, E)$ be a graph with non-negative vertex weights $w: V\rightarrow \R_{\ge 0}$ and let $W\ge 0$. Then, there exists a subset $U\subseteq V$ such that $G[U]$ is a $2$-pseudotree with $\sum_{u\in U} w(u)\le W$ if and only if there exists a pair of edges $e_1=u_1v_1, e_2=u_2v_2$ in $G$ such that there exists a Steiner tree $H=(V_H, E_H)$ in the graph $G':=G-\{e_1, e_2\}$ for  terminal set $S:=\{u_1, v_1, u_2, v_2\}$ with $\sum_{u\in V_H}w(u)\le W$.
\end{proposition}
\begin{proof}
We will say that a graph $T=(V_T, E_T)$ is a \emph{minimal} $2$-pseudotree if it is connected and has exactly $|V_T|+1$ edges. Equivalently, $T$ has exactly $2$ edges in addition to a spanning tree. We observe that for a subset $U\subseteq V$, the subgraph $G[U]$ is a $2$-pseudotree if and only if there exists a subgraph $T=(U, E_T)$ of $G[U]$ such that $T$ is a minimal $2$-pseudotree. Hence, it suffices to show that there exists a subset $U\subseteq V$ such that $G[U]$ has a subgraph that is a minimal $2$-pseudotree with $\sum_{u\in U} w(u)\le W$ if and only if there exists a pair of edges $e_1=u_1v_1, e_2=u_2v_2$ in $G$ such that there exists a Steiner tree $H=(V_H, E_H)$ in the graph $G':=G-\{e_1, e_2\}$ for  terminal set $S:=\{u_1, v_1, u_2, v_2\}$ with $\sum_{u\in V_H}w(u)\le W$. We prove this statement now. 

Let $U\subseteq V$ such that $G[U]$ contains a subgraph $T=(U, E_T)$ that is a minimal $2$-pseudotree with $\sum_{u\in U} w(u)\le W$.  
Then, there exists a pair of edges $e_1=u_1v_2$ and $e_2=u_2v_2$ in $T$ such that the subgraph $H:=T-\{e_1, e_2\}$ is acyclic, connected, and is a subgraph of $T$. In particular, we have that $H$ is acyclic, connected, and is a subgraph of $G':=G-\{e_1, e_2\}$ with $S=\{u_1, v_1, u_2, v_2\}\subseteq U=V(H)$. Hence, for the pair of edges $e_1, e_2$ in $G$, we have that $H$ is a Steiner tree in $G'$ for terminal set $S$ with $\sum_{u\in V(H)}w(u)=\sum_{u\in U} w(u) \le W$. 

Next, let $e_1=u_1v_2, e_2 = u_2v_2$ be a pair of edges in $G$ such that there exists a Steiner tree $H=(V_H, E_H)$ in the graph $G':=G-\{e_1, e_2\}$ for terminal set $S=\{u_1, v_1, u_2, v_2\}$ with $\sum_{u\in V_H}w(u)\le W$. Consider the graph $T:=H+\{e_1, e_2\}$. Then, $T$ is a minimal $2$-pseudotree. The vertex set of $T$ is $U:=V_H$ and $T$ is a subgraph of $G[U]$. Hence, we have a subset $U\subseteq V$ such that $G[U]$ has a subgraph that is a minimal $2$-pseudotree with $\sum_{u\in U}w(u)\le W$. 
\end{proof}

We now restate and prove \Cref{thm:MC2PT-polytime:main}.

\NWtwoPTPolytime
\begin{proof}
    It suffices to show that MC2PT can be solved in polynomial time. Let $G=(V,E)$ with vertex weights $w:V\rightarrow \R_{\ge 0}$ be the input instance of MC2PT.
    Consider the following algorithm: for all pairs of edges $e_1 = u_1v_1$ and $e_2 = u_2v_2$ in $E$, use the algorithm guaranteed by \Cref{prop:NWST-polytime} to solve NWST on the graph $G - \{e_1, e_2\}$ with vertex weights $w$ for terminal set $\{u_1, v_1, u_2, v_2\}$, and return the minimum weight solution over all instances. The correctness of the algorithm follows from \Cref{prop:MC2PT-to-NWST}. We now analyze the runtime of the algorithm. There are $O(|V|^2)$ pairs of edges to enumerate. Furthermore, for each pair of edges, the associated NWST instance can be solved in polynomial time using the algorithm from \Cref{prop:NWST-polytime} since each such instance has only four terminals. Thus, the algorithm runs in polynomial time.
\end{proof}

\section{Feedback Vertex Set}

\subsection{Polynomial-sized formulation of the Cycle Cover polyhedron}
\label{sec:cycle-cover-poly-sized-formulation}
In this section, we show that $\cyclecoverpolyhedron(G)$ can be expressed using polynomial number of variables and constraints. 
This is folklore, but we include its proof for the sake of completeness. For a graph $G=(V, E)$, we define $E':=\{(e, s), (e, t):\ \exists \text{ cycle in $G$ containing } e = st \in E\}$ and consider
\begin{equation}\label{eqn:distance-cycle-cover}
    \distancebasedcyclecoverpolyhedron(G) := \left\{ (x,d)\in \R^V_{\ge 0} \times \R^{V\times E'}_{\ge 0}:\ \begin{array}{l}
{d_s^{(e, s)} = 0} \hfill {\qquad \forall (e=st, s)\in E'} \\
 {d_t^{(e, s)} + x_s \ge 1} \hfill  \forall (e=st, s)\in E' \\
 {d_a^{(e,s)} + x_b \ge d_b^{(e, s)} }\hfill \qquad \forall (e, s)\in E', \ ab\in E - \{e\}
  \end{array}\right\}.  
\end{equation}

\begin{lemma}\label{lemma:cycle-cover-poly-sized-lp}
For every graph $G=(V, E)$, the polyhedron $\cyclecoverpolyhedron(G)$ is the projection of the polyhedron $\distancebasedcyclecoverpolyhedron(G)$ to $x$ variables. Consequently, $\cyclecoverpolyhedron(G)$ admits a polynomial-sized description. 
\end{lemma}
\begin{proof}
Let $\projecteddistancebasedcyclecoverpolyhedron(G)$ be the projection of $\distancebasedcyclecoverpolyhedron(G)$ to the $x$ variables. We prove that $\projecteddistancebasedcyclecoverpolyhedron(G)=\cyclecoverpolyhedron(G)$ by showing inclusion in both directions. 

First, we show that $\projecteddistancebasedcyclecoverpolyhedron(G)\subseteq \cyclecoverpolyhedron(G)$. Let $x\in \projecteddistancebasedcyclecoverpolyhedron(G)$. Let $d\in \R^{V\times E'}_{\ge 0}$ be a vector such that $(x,d)\in \distancebasedcyclecoverpolyhedron(G)$. It suffices to show that for every cycle $C$ in $G$, we have that $\sum_{u\in V(C)}x_u\ge 1$. Let $C$ be a cycle in $G$. Fix an edge $e=st\in E(C)$. By definition of the set $E'$, we have that $(e, s)\in E'$. Consequently, $d^{(e,s)}_s=0$ and $d^{(e,s)}_t + x_s\ge 1$ by the first two constraints of $\distancebasedcyclecoverpolyhedron(G)$. Let $a_1=s, a_2, a_3, \ldots, a_{r-1}, a_r=t$ be the ordered vertices along the $s-t$ path induced by $E(C) - e$. Then, for every $i\in [r-1]$, we have that $d^{(e,s)}_{a_i}+x_{a_{i+1}}\ge d^{(e,s)}_{a_{i+1}}$ by the third constraint of $\distancebasedcyclecoverpolyhedron(G)$. Adding all of these constraints gives us that $\sum_{i=2}^r x_{a_i}\ge d^{(e,s)}_{a_r}=d^{(e,s)}_t\ge 1-x_{s}=1-x_{a_1}$. Thus, the constraint $\sum_{i=1}^r x_{a_i}\ge 1$ holds. 

Next, we show that $\projecteddistancebasedcyclecoverpolyhedron(G)\supseteq \cyclecoverpolyhedron(G)$. Let $x\in \cyclecoverpolyhedron(G)$. Let $d\in \R^{V\times E'}_{\ge 0}$ be defined as follows: for each $(e=st, s)\in E'$, let 
\[
d^{(e,s)}_u:=\min\left\{\sum_{v\in P-\{s\}}x_v:\ P \text{ is a path from $s$ to $u$ 
and $P\neq \{s, t\}$
}\right\}.
\]
We show that $(x, d)\in \distancebasedcyclecoverpolyhedron(G)$. We note that the vector $(x,d)$ is non-negative by definition. Let $(e=st, s)\in E'$. Then,  $d^{(e, s)}_s=0$ by definition. 
Moreover, $d^{(e,s)}_t+x_s\ge 1$ holds because of the following reasoning: let $P$ be a path from $s$ to $t$ such that $d^{(e,s)}_t=\sum_{v\in P-\{s\}}x_v$ and $P\neq \{s, t\}$. Then, $P$ concatenated with the edge $st$ forms a cycle $C$ and $d^{(e,s)}_t + x_s = x_s+\sum_{v\in P-\{s\}}x_v = \sum_{u\in V(C)}x_u \ge 1$. 
Finally, the inequality $d^{(e,s)}_a + x_b \ge d^{(e,s)}_b$ for every edge $ab\in E$ holds because of the following reasoning: for an edge $ab\in E$, let $P$ be a path from $s$ to $a$ such that $d^{(e,s)}_a=\sum_{v\in P-\{s\}}x_v$. Then, $P$ concatenated with $b$ is a path from $s$ to $b$ and consequently, $d^{(e,s)}_b\le \sum_{v\in P+\{b\}-\{s\}}x_v$ by definition.  
\end{proof}

\subsection{Chekuri-Madan formulation}
\label{sec:CM-formulation}
In this section, we show that the polynomial-sized ILP for FVS given by Chekuri and Madan 
is such that its LP-relaxation has integrality gap at most $2$. Chekuri and Madan formulated a polynomial-sized ILP for the more general problem of Subset Feedback Vertex Set (SFVS). SFVS is defined as follows: The input is a graph $G=(V, E)$ with non-negative vertex costs $c: V\rightarrow \R_{\ge 0}$ along with a  set of terminals $T\subseteq V$. A subset $U$ of non-terminals is said to be a subset feedback vertex set if $G-U$ has no cycle containing a vertex of $T$. 
The goal is to find a least cost subset feedback vertex set. In Section \ref{sec:CM-LP}, we discuss Chekuri-Madan's ILP for SFVS (henceforth referred to as (\textsc{SFVS-IP: CM})) with some background and notation. In Section \ref{sec:reduction-from-FVS-to-SFVS}, we discuss a reduction from FVS to SFVS satisfying certain technical properties. In Section \ref{sec:cm-lp-integrality-gap}, we  show that the integrality gap of the LP-relaxation of (\textsc{SFVS-IP: CM}) specialized for FVS (henceforth referred to as the (\textsc{FVS-LP: CM})) is at most $2$.

\subsubsection{Chekuri-Madan's ILP for SFVS}
\label{sec:CM-LP}

In this section, we discuss the (\textsc{SFVS-IP: CM}).
Let $H = (V_H, E_H)$ denote the input graph, $S_H\subseteq V_H$ denote the terminal set, and $c_H:V_H\rightarrow \mathbb{R}_{\ge 0}$ denote the vertex cost function. We will say that a cycle is \emph{interesting} if it contains a terminal. Let $\calC_H$ denote the set of interesting cycles in the graph $H$.

Henceforth, we will assume that a SFVS of finite cost exists in the graph $H$. Furthermore, we will assume without loss of generality that the instance $H$ satisfies the following properties (see \cite{chekuri-madan16} for a justification of these assumptions): 
\begin{enumerate}[label=(\roman*)]
\item the graph $H$ is connected, 
\item each terminal $s_i \in S_H$
has infinite cost, is a degree two vertex with
neighbors $a_i, b_i$ having infinite cost, 
\item no two terminals
are connected by an edge or share a neighbor, 
\item there
exists a special non-terminal degree one vertex $r\in V$ with
infinite cost, and 
\item each interesting cycle contains at
least two terminals.
\end{enumerate} 
We refer to properties (i)-(v) as CM-properties.

\paragraph{Notation.} Throughout this section we will use the following notation: We let $k = |S_H|$ denote the number of terminals. We refer to the set $P_i := \{a_i, b_i\}$ as the \emph{pivot set} for the terminal $s_i$, and the set $P_H = \cup_{i \in [k]}P_i$ as the set of all pivots. We define the set 
\[
\ell_H(u) = \{i\in [k]: \text{$s_i$ is reachable from $u$  via a path not containing any other terminals}\} \cup \{k+1\} 
\]
as the set of \emph{labels} for each vertex $u \in V_H \backslash S$. We denote $\Tilde{V}_H$ as the set of vertices in $V_H$ that are not in $S_H \uplus P_H \uplus \{r\}$. We denote $\Tilde{E}_H$ as the set of edges in $E_H$ that are not incident to any terminals and the set $E'_H$ as the set of edges incident to terminals. We will refer to $E'_H$ as special edges. 

\paragraph{(\textsc{SFVS-IP: CM}).} 
See \Cref{fig:CM-lp-SFVS} for 
the labelling-based (\textsc{SFVS-IP: CM}).
In this ILP, we have vertex variables $x_u$ for each $u \in \Tilde{V}$ indicating whether the vertex $u$ is in the SFVS $F$, and vertex labelling variables $z_{u, i}$ for each $u \in \Tilde{V}\cup P_H, i\in [k]$ indicating whether vertex $u$ receives label $i$. We note that \Cref{fig:CM-lp-SFVS} simplifies the exposition of the ILP from \cite{chekuri-madan16} by explicitly substituting the values for variables whose values are directly set by constraints, i.e. $x_u = 0$ for $u \in S \cup\{r\}$, $z_{s_i, i} = 1$ for $i \in [k]$, and $z_{r, k+1} = 1$. Moreover, it slightly strengthens the first constraint (Chekuri-Madan used a slightly weaker constraint: $x_u + \sum_{i\in [k+1]}z_{u, i}=1$).

The constraints can be interpreted as follows: Let $F$ be a minimal finite cost subset feedback vertex set of $H$. 
Constraint (1) says that if a vertex $u \in \Tilde{V}_H$ is not in $F$, then it must receive exactly one of the labels in $\ell_H(u)$. Constraint (2) says that exactly one of $a_i$ and $b_i$ can receive the label $i$ for each $i \in [k]$. Constraint (3) is a spreading constraint that says that both end vertices of every non-special edge in $H - F$ must receive the same label. Constraint (4) is a cycle covering constraint and says that the solution $F$ must intersect every interesting cycle in at least one vertex. Constraint (5) ensures that the set $F$ contains no pivot vertices.

We now briefly address why a labeling satisfying the LP constraints must exist for every minimal finite weight subset feedback vertex set $F$ (see \cite{chekuri-madan16} for further details). We note that none of the vertices of $S\cup P\cup \{r\}$ are in $F$ as these have infinite cost. Let the graph $H' = (V_{H'}, E_{H'})$ be defined as $H' = H - F$. For simplicity, we first assume that the graph $H'$ is connected and then address the more general case when $H'$ is not connected. Since there is no cycle containing any terminal in $H'$, each terminal is a cut vertex in $H'$. Consider the \emph{block-cut-vertex tree} decomposition $T'$ of $H'$ --- (refer to \cite{west_introduction_2000} for details on block-cut-vertex tree decompositions). We label a vertex $u$ of $H'$ as $i \in [k]$ if terminal $s_i$ is the first terminal encountered on any path from vertex $u$ to the root $r$ in the graph $H'$. If no terminals are encountered, then we label the vertex as $k+1$. We note that such a labelling is well-defined as $T'$ is the block-cut-vertex tree and the set $S$ are cut vertices. In this labelling, we observe that every non-special edge $uv \in \Tilde{E}_{H'}$ has the property that both end points receive the same label. Finally, in the case where $H'$ is not connected, we
can justify the existence of such a labeling by picking an arbitrary
non-terminal vertex from each component, imagining a dummy edge connecting it to the root $r$, and considering the labeling corresponding to block-cut-vertex tree of this (implicit) graph.

\begin{figure}
    \centering
    \begin{mdframed}
    \textbf{Subset Feedback Vertex Set:}
    \begin{align*}
        \text{min}\qquad&\displaystyle\sum\limits_{u \in \Tilde{V}_H} c_{u}x_{u}&\\
        \text{s.t.}\qquad&&\\
        (1)\ \ &x_u + \sum_{i \in \ell_H(u)}z_{u,i} = 1& \forall\ u\in \Tilde{V}_H\cup P_H\\
        (2)\ \ &z_{a_i, i} + z_{b_i, i} = 1& \forall\ i \in [k]\\
        (3) \ \ &x_u + z_{u,i} - z_{v, i} \geq 0&     \forall\ uv\in \Tilde{E}_H,  \ i\in [k+1] \\
        (4)\ \ &\sum_{u\in C}x_u \geq 1& \forall\ C\in\calC_H\\
        (5)\ \ &x_u = 0&\forall\ u \in P_H\\
        (6)\ \ &z_{u, i} \geq 0&\forall\  \text{  $u \in \Tilde{V}_H \cup P_H, i\in [k+1]$}\\
        (7)\ \ &x_u \geq 0&\forall\  u \in \Tilde{V}_H\\
        (8)\ \ &z_{u, i} \in \Z&\forall\  \text{  $u \in \Tilde{V}_H \cup P_H, i\in [k+1]$}\\
        (9)\ \ &x_u \in \Z &\forall\  u \in \Tilde{V}_H
    \end{align*}
    \end{mdframed}
    \caption{
    (\textsc{SFVS-IP: CM})
    }
  \label{fig:CM-lp-SFVS}
\end{figure}

\subsubsection{Reduction from FVS to SFVS satisfying the CM-properties}\label{sec:reduction-from-FVS-to-SFVS}
FVS can be seen as a special case of SFVS where every vertex in the graph is a terminal. However, this simple reduction does not result in an SFVS instance satisfying the CM-properties (see properties (i)--(v) mentioned in the first paragraph of Section \ref{sec:CM-LP}). In this section, we show how to construct 
an SFVS instance from an FVS instance such that the SFVS instance satisfies the required properties. Let $G= (V_G, E_G)$ with vertex costs $c_G:V_G\rightarrow \R_{\ge 0}$ be the FVS instance. We will construct 
an SFVS instance $H = (V_H, E_H)$ with terminals $S_H$ and vertex costs  $c_H:V_{H}\rightarrow\mathbb{R}_{\ge 0}$ satisfying the required properties. We may assume that the FVS instance $G$ has an infinite-weighted degree-1 root vertex $r$, satisfying property (iv) since adding such a vertex does not change the set of minimal feasible solutions. We construct $H$ as follows:

\begin{enumerate}
    \item For each edge $e = uv \in E_G$, we subdivide the edge $e$ so that the edge becomes a path $u, a_e, s_e, b_e, v$ of length $4$. 
    \item We let $S_H = \bigcup_{e\in E(G)} \{s_e\}$ denote the terminal set. Thus, the set $P_e := \{a_e, b_e\}$ is the set of pivot vertices for each terminal $s_e \in S_H$, and the set $P_H := \bigcup_{e\in E(G)}P_e$ is the set of all pivots of the graph $H$.
    \item We define the cost function as $c_H(u) = c_G(u)$ if $u \in V_G$, and $c_H(u) = \infty$ otherwise (i.e., if $u \in S_H\cup P_H$).
\end{enumerate}

We observe that the graph $H$ satisfies properties (i)-(v) assumed by the Chekuri-Madan LP. The next proposition says that solving FVS in the graph $G$ with cost function $c_G$ is equivalent to solving SFVS in the graph $H$ with terminal set $S_H$ with cost function $c_H$. The proposition follows by construction.

\begin{proposition}\label{prop:reduction-FVS-to-SFVS}
Let $F\subseteq V_G$. The set $F$ is a feedback vertex set for $G$ if and only if the set $F$ is a subset feedback vertex set for the graph $H$ with terminal set $S_H$. Furthermore, the cost of $F$ in both graphs is the same, i.e. $c_G(F) = c_{H}(F)$.
\end{proposition}

\cite{chekuri-madan16} showed that the LP-relaxation of (\textsc{SFVS-IP: CM}) given in \Cref{fig:CM-lp-SFVS} has integrality gap at most $13$. It immediately follows that the integrality gap of that LP-relaxation on the instance $H$ constructed as above is also at most $13$. 
In the subsequent section, we tighten this integrality gap to $2$ for instances $H$ constructed as above.

\subsubsection{Bounding the integrality gap}\label{sec:cm-lp-integrality-gap}
In this section, we show that the LP-relaxation of the (\textsc{SFVS-IP: CM}) given in \Cref{fig:CM-lp-SFVS} for the instance $H$ constructed as given in the reduction in Section \ref{sec:reduction-from-FVS-to-SFVS} has integrality gap at 
 most $2$. We prove this by showing that the LP constraints imply orientation and cycle cover constraints. 
Consequently, the integrality gap of the LP-relaxation is at most that of $\min\{\sum_{u\in V} c_ux_u: x\in \projectedorientationpolyhedron(G)\cap \cyclecoverpolyhedron(G)\}$. By \Cref{lemma:integrality-gap-orient-intersect-cycle-cover}, it follows that the integrality gap of the LP-relaxation is at most $2$. 



We begin by simplifying the (\textsc{SFVS-IP: CM}) given in \Cref{fig:CM-lp-SFVS} for SFVS instances $H$ that are constructed from FVS instances as obtained from the reduction in Section \ref{sec:reduction-from-FVS-to-SFVS}. 
We make two observations that will help us in the simplification.
The first observation is that each edge $e \in E_G$ corresponds to a unique terminal $s_e \in V_H$. Thus, we may replace the set $[k]$ which indexed the set of terminals in the ILP of \Cref{fig:CM-lp-SFVS} with the set $E_G$. For ease of exposition, we denote $e_r$ by the label $k+1$.
We now give the second observation. Let $u \in V_G$ be arbitrary. Since each edge $e \in \delta_G(u)$ has been subdivided in the graph $H$ and contains a terminal $s_e$, the only terminals that $u$ can reach via a path that does not contain any other terminal is exactly $\{s_e: e \in \delta_G(u)\}$. For an arbitrary pivot vertex $u \in P_H$, let $g(u)$ denote the unique non-terminal neighbor of $u$ in $H$. Then, the only terminals that $u$ can reach via a path that does not contain any other terminal is exactly $\{s_e: e \in \delta_G(g(u))\}$. We summarize these observations in the next proposition.


\begin{proposition}\label{prop:cm-lp-fvs-stronger-label-set}
We have that
\begin{enumerate}
    \item For each vertex $u \in V_G$, we have that $\ell_H(u) = \delta_G(u)\cup \{e_r\}$;
    \item For each pivot vertex $u \in P$, we have that $\ell_H(u) = \delta_G(g(u))\cup \{e_r\}$.
\end{enumerate}
\end{proposition}

Using \Cref{prop:cm-lp-fvs-stronger-label-set} and the first observation mentioned above, we simplify the (\textsc{SFVS-IP: CM}) in \Cref{fig:CM-lp-SFVS} for the SFVS instance $H$ and write its LP-relaxation (\textsc{FVS-LP: CM}) in \Cref{fig:CM-lp-FVS}. 
\begin{figure}[H]
    \centering
    \begin{mdframed}
    \textbf{LP-relaxation of (\textsc{SFVS-IP: CM}) for graph $H$:}
    \begin{align*}
        \text{min}\qquad&\displaystyle\sum\limits_{u \in \Tilde{V}_H} c_{u}x_{u}&\\
        \text{s.t.}\qquad&&\\
        (1)\ \ &x_u + \sum_{e \in \ell_H(u)}z_{u,e} = 1& \forall u\in \Tilde{V}_H\cup P_H\\
        (2)\ \ &z_{a_e, e} + z_{b_e, e} = 1& \forall e \in E_G\\
        (3) \ \ &x_u + z_{u,e} - z_{v, e} \geq 0&  \forall   uv\in \Tilde{E}_H, \  e\in E_G\cup\{e_r\} \\
        (4)\ \ &\sum_{u\in C}x_u \geq 1& \ \forall C\in\calC_H\\
        (5)\ \ &x_u = 0&\  \text{ $\forall u \in P_H$}\\
        (6)\ \ &z_{u, e} \geq 0&\  \text{  $\forall u \in \Tilde{V}_H\cup P_H, e\in E_{G}\cup \{e_r\}$}\\
        (7)\ \ &x_u \geq 0&\  \text{ $\forall u \in \Tilde{V}_H$}
    \end{align*}
    \end{mdframed}
    \caption{
    (\textsc{FVS-LP: CM})
    }
  \label{fig:CM-lp-FVS}
\end{figure}

We now show that the constraints of the LP in \Cref{fig:CM-lp-FVS} imply the orientation and cycle cover constraints for $G$. 

\begin{lemma}\label{lem:cm-implies-orientation}
Let $(\overline{x}, \overline{z})$ be a feasible solution to (\textsc{FVS-LP: CM}) from \Cref{fig:CM-lp-FVS}. Then,
\begin{enumerate}
\item $\overline{x}\in \cyclecoverpolyhedron(G)$ and 
\item  
$(\overline{x}, \overline{y}) \in \orientationpolyhedron(G)$, where $\overline{y}_{e, u}=\overline{z}_{u, e}$ for all $e\in \delta(u), u\in V_G$. 
\end{enumerate}
\end{lemma}
\begin{proof}
We have that $\overline{x}\in \cyclecoverpolyhedron(G)$ owing to inequality (4). We now show that $(\overline{x}, \overline{y}) \in \orientationpolyhedron(G)$.
We have that $\overline{x}_u\ge 0$ for all $u\in V_G$ and $\overline{y}_{e, u}\ge 0$ for all $e\in \delta_G(u), u\in V_G$. We now show that $\overline{x}_u + \sum_{e\in \delta_G(u)}\overline{y}_{e,u}\le 1$ for every $u\in V_G$. Let $u\in V$. We have that 
\[
\overline{x}_u + \sum_{e\in \delta_G(u)}\overline{y}_{e,u}=\overline{x}_u + \sum_{e\in \ell_H(u)}\overline{z}_{u, e}-\overline{z}_{u, e_r}= 1-\overline{z}_{u, e_r}\le 1,
\]
where the last inequality holds since $\overline{z}_{u, e_r}\ge 0$. 

Finally, we show that $\overline{x}_u + \overline{x}_v + \overline{y}_{e, u} +\overline{y}_{e,v}\ge 1$ for every $e=uv\in E_G$. Let $e=uv\in E_G$. Let $u, a_e, s_e, b_e,v$ be the path in the graph $H$ corresponding to the subdivided edge $e$. Then, we have that 
\begin{align*}
\overline{x}_u + \overline{x}_v + \overline{y}_{e, u} + \overline{y}_{e, v}& =
    \overline{x}_u + \overline{x}_v + \overline{z}_{u, e} + \overline{z}_{v, e}&\\
    &= (\overline{x}_u + \overline{x}_v + \overline{z}_{u, e} + \overline{z}_{v, e}) + (\overline{z}_{a_e, e} + \overline{z}_{b_e, e}) - (\overline{z}_{a_e, e} + \overline{z}_{b_e, e}) &\\
    &= 1 + ( \overline{x}_u + \overline{z}_{u,e} - \overline{z}_{a_e, e})+ ( \overline{x}_v + \overline{z}_{v,e} - \overline{z}_{b_e, e})&\\
    & \geq 1.&
\end{align*}
Here, the third equality follows from constraint (2) and the final inequality follows from constraint (3) of (\textsc{FVS-LP: CM}) in \Cref{fig:CM-lp-FVS}.
\end{proof}

\Cref{lem:cm-implies-orientation}  implies that $\min\{\sum_{u\in V}c_u x_u: x\in \projectedorientationpolyhedron(G)\cap \cyclecoverpolyhedron(G)\}$ is a relaxation of (\textsc{FVS-LP: CM}) from \Cref{fig:CM-lp-FVS}.  \Cref{lemma:integrality-gap-orient-intersect-cycle-cover} tells us that the integrality gap of $\min\{\sum_{u\in V}c_u x_u: x\in \projectedorientationpolyhedron(G)\cap \cyclecoverpolyhedron(G)\}$ is at most $2$. Consequently, the integrality gap of (\textsc{FVS-LP: CM}) from \Cref{fig:CM-lp-FVS} is also at most $2$. We note that (\textsc{FVS-LP: CM}) from \Cref{fig:CM-lp-FVS} can be converted to a polynomial-sized LP by replacing inequality (4) using a polynomial-sized description of cycle cover constraints as given in \Cref{lemma:cycle-cover-poly-sized-lp}.

\subsection{Orientation formulation}\label{sec:orientation-formulation}
In this section, we give another polynomial-sized ILP for FVS whose LP-relaxation has integrality gap at most $2$. This ILP is based only on orientation constraints and does not require cycle cover constraints (in contrast to the ILP presented in Section \ref{sec:orientation-and-cycle-cover-formulation}). We consider the following formulation and denote it as the (\textsc{FVS-IP: Complete-Orient}):
\begin{align}
\min & \sum_{u\in V} c_u x_u \nonumber\\
x_v + x_w + y^{f}_{e, v}+y^f_{e, w}&\ge 1\ \forall\ e=vw\in E, f\in E \label{constraint:orientation}\\
x_v+\sum_{e=vw\in E}y^f_{e,w}&\ge 1\ \forall\ v\in V, f\in E \label{constraint:load}\\
\sum_{v\in V-\{a,b\}}x_v + \sum_{e=vw\in E-\{f\}}(y^f_{e,w}+y^f_{e,v})&\le |V|-2\ \forall\ f=ab\in E \label{constraint:cumulative}\\
x_u&\ge 0\ \forall u\in V \label{constraint:x-nonnegative}\\
y^f_{e,v}&\ge 0\ \forall e\in\delta(v), v\in V, f\in E \label{constraint:y-nonnegative}\\
x_u&\in \Z\ \forall u\in V \nonumber\\
y^f_{e,v}&\in \Z\ \forall e\in\delta(v), v\in V, f\in E. \nonumber
\end{align}

\Cref{thm:poly-sized-LP-with-integrality-gap-atmost-2-for-FVS} follows from the following lemma. 

\begin{lemma}\label{lemma:strong-density-orientation-lp}
For an input graph $G=(V, E)$ with non-negative vertex costs $c:V\rightarrow \R_{\ge 0}$, 
(\textsc{FVS-IP: Complete-Orient}) is an integer linear programming formulation for FVS. Moreover, its LP-relaxation has integrality gap at most $2$.  
\end{lemma}
\begin{proof}
Let $x$ be the indicator vector of a feedback vertex set $S$ of $G$. 
We show that there exists a binary vector $y$ satisfying the constraints. Let $T$ be a spanning tree of $G$ containing all edges of $G$ that are not incident to vertices in $S$. 
Fix an edge $f=ab\in E$. 
For each vertex $v\in V\setminus (S\cup \{a, b\})$, set $y^f_{e,w}=1$ where $e=vw$ is the first edge on the unique path from $v$ to $a$ in $T$. Moreover, set $y^f_{f,a}=y^f_{f,b}=1$ and all other $y^f$ variables to $0$. We prove that this choice of $y$ satisfies all constraints. 

For an edge $e=vw\in E$, if $|S\cap \{v,w\}|\ge 1$, then $x_v+x_w\ge 1$ and if $v, w\not\in S$, then $e\in T$ and consequently, either $y_{e,v}^f=1$ or $y^f_{e,w}=1$ and hence, the constraint $x_v + x_w + y^{f}_{e, v}+y^f_{e, w}\ge 1$ holds. For a vertex $v\in V$, if $v\in S$, then $x_v\ge 1$ and if $v\in V-S$, then there exists an edge $e=vw\in T$ and hence, $y^f_{e,w}=1$ and consequently, the constraint $x_v+\sum_{e=vw\in E}y^f_{e,w}\ge 1$ holds. It remains to show that the constraint $\sum_{v\in V-\{a,b\}}x_v + \sum_{e=vw\in E-\{f\}}(y^f_{e,w}+y^f_{e,v})\le |V|-2$ holds. By the choice of $x$, it suffices to show that $|S\setminus \{a, b\}| + \sum_{e=vw\in T-\{f\}}(y^f_{e,w}+y^f_{e,v})\le |V|-2$ holds. By the choice of $y^f$, it suffices to show that $|S\setminus \{a, b\}| + \sum_{e=vw\in T}(y^f_{e,w}+y^f_{e,v})\le |V|$. Simplifying this further, it suffices to show that 
\[
\sum_{e=vw\in T}(y^f_{e,w}+y^f_{e,v})\le |(V-S)\cup\{a, b\}|.
\]
We note that $\sum_{e=vw\in T}(y^f_{e,w}+y^f_{e,v}) = \sum_{v\in V}(\sum_{e=vw\in T}y^f_{e,w})$. Hence, it suffices to show that 
\[
\sum_{v\in V}\left(\sum_{e=vw\in T}y^f_{e,w}\right) \le |(V-S)\cup\{a, b\}|.
\]
We have the following four inequalities:
\begin{align*}
\sum_{e=vw\in T}y^f_{e, w}&\le 1\ \forall v\in V-(S\cup \{a, b\}),\\
\sum_{e=vw\in T}y^f_{e, w}&= 0\ \forall v\in S-\{a, b\},\\
\sum_{e=aw\in T}y^f_{e, w}&= y^f_{e,b}=1, \text{ and}\\
\sum_{e=bw\in T}y^f_{e, w}&= y^f_{e,a}=1.
\end{align*}
Hence, $\sum_{v\in V}\left(\sum_{e=vw\in T}y^f_{e,w}\right) \le |(V-S)\cup\{a, b\}|$. 

Next, we show that an integral solution to the system of inequalities (\ref{constraint:orientation}), (\ref{constraint:load}), (\ref{constraint:cumulative}), (\ref{constraint:x-nonnegative}), and (\ref{constraint:y-nonnegative}) is the indicator vector of a feedback vertex set. 
We will use the following claim. 
\begin{claim}\label{claim:orientation-constraints-imply-strong-density}
Constraints (\ref{constraint:orientation}), (\ref{constraint:load}), (\ref{constraint:cumulative}), (\ref{constraint:x-nonnegative}), and (\ref{constraint:y-nonnegative}) imply the strong density constraints. 
\end{claim}
\begin{proof}
Let $(x, y)$ be a solution satisfying constraints (\ref{constraint:orientation}), (\ref{constraint:load}), (\ref{constraint:cumulative}), (\ref{constraint:x-nonnegative}), and (\ref{constraint:y-nonnegative}). Let $S\subseteq V$ such that $E[S]\neq \emptyset$. We will show that 
\[
\sum_{u\in S}(d_S(u)-1)x_u \ge |E[S]|-|S|+1.
\]
Equivalently, we will show that 
\[
|E[S]|-|S| - \sum_{u\in S}(d_S(u)-1)x_u \le -1.
\]
Let us rewrite the LHS as
\[
|E[S]|-|S| - \sum_{u\in S}(d_S(u)-1)x_u = \sum_{vw\in E[S]} (1-x_v-x_w) + \sum_{v\in S}(x_v-1).
\]
Let us inspect each sum separately. Fix an arbitrary $f=ab\in E[S]$. We have that 
\begin{align*}
\sum_{vw\in E[S]} (1-x_v-x_w) 
&= \sum_{vw\in E[S]} \left(1-x_v-x_w-y^f_{e,w}-y^f_{e,v}\right) + \sum_{vw\in E[S]} \left(y^f_{e,v} + y^f_{e,w}\right)\\
&\le 1-x_a-x_b-y^f_{f,a}-y^f_{f,b}+\sum_{vw\in E[S]} \left(y^f_{e,v} + y^f_{e,w}\right)
\end{align*}
where the inequality follows from (\ref{constraint:orientation}) and the fact that $f=ab\in E[S]$. We also have that 
\begin{align*}
\sum_{v\in S}(x_v -1) 
&= \sum_{v\in S} \left(x_v -1 +\sum_{e=vw\in E}y^f_{e, w}\right) - \sum_{v\in S}\sum_{e=vw\in E} y^f_{e,w}\\
&\le \sum_{v\in V} \left(x_v -1 +\sum_{e=vw\in E}y^f_{e, w}\right) - \sum_{v\in S}\sum_{e=vw\in E} y^f_{e,w}
\end{align*}
where the last inequality follows from (\ref{constraint:load}). Before adding the two upper bounds, we also note that 
\[
\sum_{vw\in E[S]} \left(y^f_{e,v} + y^f_{e,w}\right) - 
\sum_{v\in S}\sum_{e=vw\in E} y^f_{e,w}
\le 0
\]
by (\ref{constraint:y-nonnegative}). 
Hence, 
\begin{align*}
\sum_{vw\in E[S]} (1-x_v-x_w) + \sum_{v\in S}(x_v -1) 
&\le 1-x_a-x_b-y^f_{f,a}-y^f_{f,b}+\sum_{v\in V} \left(x_v -1 +\sum_{e=vw\in E}y^f_{e, w}\right) \\
&= 1 - |V| + \sum_{v\in V-\{a, b\}}x_v + \sum_{vw\in E-\{f\}}(y^f_{e,v}+y^f_{e,w})\\
&\le -1
\end{align*}
where the last inequality follows from (\ref{constraint:cumulative}). 
\end{proof}
By the results of Chudak, Goemans, Hochbaum, and Williamson \cite{CHUDAK1998111}, every integral solution satisfying strong density constraints corresponds to the indicator vector of a feedback vertex set. Thus, by Claim \ref{claim:orientation-constraints-imply-strong-density}, an integral solution to the system of inequalities (\ref{constraint:orientation}), (\ref{constraint:load}), (\ref{constraint:cumulative}), (\ref{constraint:x-nonnegative}), and (\ref{constraint:y-nonnegative}) is the indicator vector of a feedback vertex set. This completes the proof that the formulation given in the lemma is indeed an integer linear programming formulation for FVS. 

Finally, we bound the integrality gap of the LP-relaxation. By the results of Chudak, Goemans, Hochbaum, and Williamson \cite{CHUDAK1998111}, we know that (\ref{FVS-IP: SD})
is an ILP formulation for FVS and the integrality gap of its LP-relaxation 
(\ref{FVS-LP: SD}) 
is at most $2$. 
We have already shown that the ILP in the lemma statement is a valid formulation for FVS. 
By Claim \ref{claim:orientation-constraints-imply-strong-density}, the LP-relaxation of the ILP in the lemma statement has integrality gap at most $2$.  
\end{proof}

\end{document}